\documentclass[journal,10pt]{IEEEtran} 
\makeatletter
\def\endthebibliography{%
	\def\@noitemerr{\@latex@warning{Empty `thebibliography' environment}}%
	\endlist
}
\makeatother

\usepackage{cite}

\usepackage[pdftex]{graphicx}

\usepackage[caption=false,font=footnotesize]{subfig}

\usepackage{amsmath}
\usepackage{amsfonts}
\usepackage{bm}
\usepackage{filecontents}
\usepackage{amssymb}
\usepackage{color}
\usepackage[table]{xcolor}
\DeclareMathOperator*{\argmin}{\arg\min}

\usepackage{epsfig}

\usepackage{url}

\usepackage{algpseudocode}
\usepackage{algorithm, tabularx}

\usepackage{lettrine}

\usepackage{lipsum}

\usepackage{siunitx}

\usepackage{soul,color}

\usepackage{mathrsfs}

\usepackage{array}

\usepackage[inline]{enumitem}

\usepackage{breqn}

\usepackage{multirow}
\usepackage{array}
\newcolumntype{L}[1]{>{\raggedright\let\newline\\\arraybackslash\hspace{0pt}}m{#1}}
\newcolumntype{C}[1]{>{\centering\let\newline\\\arraybackslash\hspace{0pt}}m{#1}}
\newcolumntype{R}[1]{>{\raggedleft\let\newline\\\arraybackslash\hspace{0pt}}m{#1}}

\usepackage{amsthm}
\newtheorem{theorem}{Theorem}

\newtheorem{lemma}[theorem]{Lemma}

\newlength{\maxwidth}
\newcommand{\algalign}[2]
{\makebox[\maxwidth][r]{$#1{}$}${}#2$}

\makeatletter
\newcommand{\multiline}[1]{%
	\begin{tabularx}{\dimexpr\linewidth-\ALG@thistlm}[t]{@{}X@{}}
		#1
	\end{tabularx}
}
\makeatother

\theoremstyle{remark}

\newtheorem{remark}{Remark}

\algdef{SE}[SUBALG]{Indent}{EndIndent}{}{\algorithmicend\ }%
\algtext*{Indent}
\algtext*{EndIndent}

\begin{document}
	
	\title{An Energy-efficient Aerial Backhaul System with Reconfigurable Intelligent Surface}

	\author{Hong-Bae Jeon,~\IEEEmembership{Graduate Student Member,~IEEE}, Sung-Ho Park,~\IEEEmembership{Member, IEEE}, Jaedon Park,\\ 
	Kaibin Huang,~\IEEEmembership{Fellow,~IEEE}, and Chan-Byoung Chae,~\IEEEmembership{Fellow,~IEEE}
			
		\thanks{This research was supported by the Agency for Defense Development, Rep. of Korea. This article was presented in part at the IEEE Globecom 2021, Madrid, Spain, December 7–11, 2021~\cite{URG}. \textit{(Corresponding Author: Chan-Byoung Chae.)}}
\thanks{H.-B. Jeon and C.-B. Chae are with the School of Integrated Technology, Yonsei University, Seoul 03722, Rep. of Korea (e-mail: \{hongbae08, cbchae\}@yonsei.ac.kr).}
\thanks{S.-H. Park is with the Datasolution, Inc., Seoul 06101, Rep. of Korea (e-mail: shpark@datasolution.kr).}
\thanks{J. Park is with the Agency for Defense Development, Daejeon 34186, Rep. of Korea (e-mail: jaedon2@add.re.kr).}
\thanks{K. Huang is with the Department of Electrical and Electronic Engineering, The University of Hong Kong, Pok Fu Lam, Hong Kong (e-mail: huangkb@eee.hku.hk).}
	}
	
	
	
	\maketitle

	\begin{abstract}
	In this paper, we propose a novel wireless architecture, mounted on a high-altitude aerial platform, which is enabled by reconfigurable intelligent surface (RIS). 
By installing RIS on the aerial platform, rich line-of-sight and full-area coverage can be achieved, thereby, overcoming the limitations of the conventional terrestrial RIS. 
We consider a scenario where a sudden increase in traffic in an urban area triggers authorities to rapidly deploy unmanned-aerial vehicle base stations (UAV-BSs) to serve the ground users. In this scenario, since the direct backhaul link from the ground source can be blocked due to several obstacles from the urban area, we propose reflecting the backhaul signal using aerial-RIS so that it successfully reaches the UAV-BSs. We jointly optimize the placement and array-partition strategies of aerial-RIS and the phases of RIS elements, which leads to an increase in energy-efficiency of every UAV-BS. We show that the complexity of our algorithm can be bounded by the quadratic order, thus implying high computational efficiency. We verify the performance of the proposed algorithm via extensive numerical evaluations and show that our method achieves an outstanding performance in terms of energy-efficiency compared to benchmark schemes.
    
	\end{abstract}

	
		\IEEEpeerreviewmaketitle
		
	\section{Introduction}
\lettrine{W}{HEN} there are issues concerning the wireless network in terms of ground base stations (BSs) malfunctioning or a massive increase in network traffic in a wireless network, a promising solution is to use unmanned-aerial-vehicles (UAVs) as temporary BSs~\cite{ruizhangmagazine, mozatut, RAFT}. The mobility and high-altitudes of UAV-BSs allow flexible deployment deployment and provide rich line-of-sight (LoS) links~\cite{a2gglobecom}. To efficiently deploy UAV-BSs, researchers have focused on designing the deployment strategies~\cite{mozaeff, lapotp}, trajectory optimization~\cite{traj, mozaD2D, TCD} and power control~\cite{Noh, mozaenergy, eemrtp}. 
 
 One of the implicit conditions for the UAV-BS deployment is that the backhaul rate should be high enough to support the traffic of the UAV-BSs~\cite{DBS}. To satisfy this requirement, the energy-efficiency of the backhaul source should also be considered~\cite{EEGC, sg}. However, providing backhaul links directly from the ground source is inefficient owing to potential blockages, which generate non-LoS (NLoS) links and significantly degrade the system's energy-efficiency. This gives rise to a critical issue for future wireless networks with higher frequencies and wider bandwidths~\cite{heath2014fiveDisrup, WM, cholens, OP, jang2016smart}.
\begin{figure*}[t]
	\begin{center}
		\includegraphics[width=1.9\columnwidth,keepaspectratio]%
		{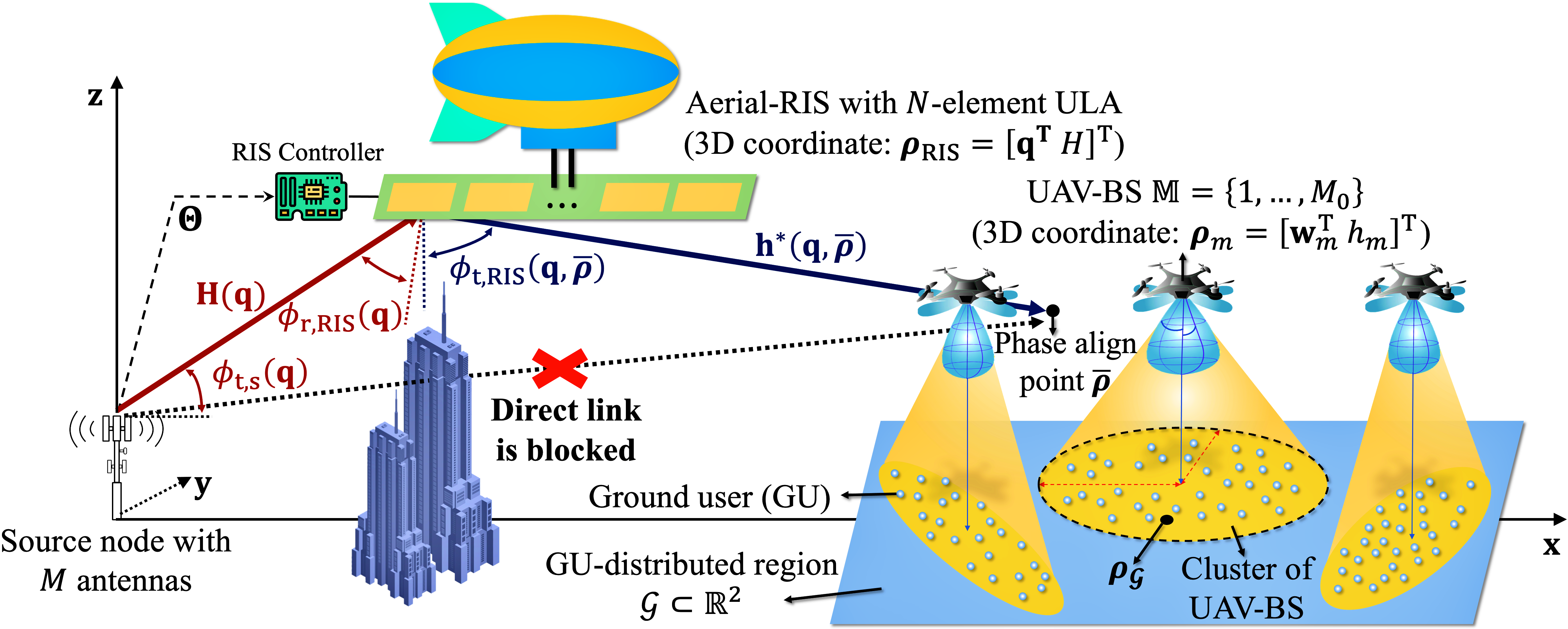}
		\caption{Proposed UAV-BS access network with aerial-RIS backhaul link. The source is equipped with $M$ directional antennas and sends the backhaul signal to $M_0$ UAV-BSs. The RIS is assumed to be implemented on a voluminous aerial platform placed at an altitude of $H$, and reflects the backhaul signal from the source with high-LoS probability.}
		\label{fig_overview}
	\end{center}
\end{figure*}

Addressing the issue motivates researchers to design reconfigurable intelligent surfaces (RISs) to create LoS links by smart reflection and to direct them in desired directions~\cite{RISA, RIST}. RIS is an artificial metasurface composed of passive reflecting elements that can adjust the amplitude and phase of a reflected signal~\cite{meta, PSM}. Moreover, since the array architecture of the RIS is passive, it features low power consumption~\cite{LingRIS}. Therefore, RISs are expected to be deployed in future wireless networks to enhance their performance. In~\cite{RISEE, JAP}, the authors improved the energy-efficiency by jointly allocating transmit power and adjusting phase shifts of RIS. In~\cite{BO, FBT, MIMORIS}, the researchers proposed various RIS transmission strategies, including a combined design of continuous transmit precoder and discrete phase shifts~\cite{BO}, single-to-multi-beam training methods by partitioning RIS array and designing their beam directions~\cite{FBT}, and constructing multiple-input-multiple-output quadrature amplitude modulation (MIMO-QAM) under the RIS hardware constraints such as discrete phase shift and phase-dependent amplitude~\cite{MIMORIS}. Moreover, maximizing the spectral efficiency is studied in~\cite{RO1, RO2, LCSI}. The prior works, however, mainly focuses on deploying RIS in a terrestrial environment that includes buildings and walls. Such an environment poses several limitations on communication performance. First of all, in urban areas with many buildings, reliable communication requires multiple reflections, which requires the installing of many RISs to avoid severe signal attenuation~\cite{RISpath}. Moreover, the terrestrial RIS can only reflect the signal from the source to the destination on the same side of the RIS, which cannot achieve isotropic reflection with $0^{\circ}\sim360^{\circ}$ arrival-angles. Although some attempts have been made to implement the RIS on the aerial platforms~\cite{SRS, UAVRIS, UPARIS}, only the 2D coverage is considered or the fixed transmit power assumed, without considering the energy-efficiency problem.


To address these issues, we, in this paper, propose an energy-efficient aerial backhaul structure~\cite{IAB} by mounting RIS on a sizable aerial platform, as motivated by the results in~\cite{UAVRIS} and illustrated in Fig.~\ref{fig_overview}. By combining the properties of the RIS with the rich-LoS capability of the aerial platform, the aerial-RIS gives rise to a favorable channel with a reduced number of reflections. 
This reduces transmit power for a given data rate and thereby increases the energy-efficiency. Moreover, since the aerial-RIS can reflect the signal in an arbitrary direction~\cite{UAVRIS}, it supports the full-3D coverage, including the receivers (UAV-BSs) in the air. Exploiting these benefits, we further propose an energy-efficiency enhancement strategy. The contributions of this paper are summarized as:
 \begin{enumerate}
 \item We propose to deploy RIS on the aerial platform to establish backhaul links with rich-LoS component to the UAV-BSs, which leverages the advantages of both a high-altitude platform and the RIS. We first adapt the maximum-ratio transmission (MRT) strategy to maximize the received signal-to-noise ratio (SNR) and thereby lower the transmit power.
 \item To achieve a high energy-efficiency, we derive a non-zero lower-bound on the source transmit power for every UAV-BS and minimize the bounds by minimizing and maximizing their numerator and denominator, respectively. It consists of achieving the Pareto-optimum using the method of global criterion, array-partition-mode selection with our proposed “reverse-waterfilling” partitioning strategy, and the exploitation of the relationship between the passive beamforming gain and the sin-angle-deviations.
 \item Through our novel framework of aerial-RIS deployments, we optimize the placement, array-partition, and phase alignment of aerial-RIS to minimize the source transmit power under the fronthaul-rate-ensuring constraints. Furthermore, we show that the computational complexity of the proposed algorithm is upper-bounded by the quadratic order, which shows a high computational efficiency. We numerically estimate the performance of the proposed algorithm in a realistic urban outdoor scenario with randomly distributed ground users and corresponding UAV-BSs. By extensive simulations, we demonstrate substantial improvements of energy-efficiency and considerable diminishment of the optimality gap by the proposed~algorithm.
 \end{enumerate}
 

The rest of the paper is organized as follows: In Section~II, the system model is described, including the target power minimization problem with the fronthaul rate constraints. 
In Section~III, we present our proposed algorithm based on maximizing the lower-bound of the transmit power in the sense of the placement of aerial-RIS, phase alignment, and the RIS array-partition. The computational complexity of the proposed algorithm is also analyzed. We present the numerical results of the proposed algorithm in Section~IV, followed by the conclusion and future works in Section~V.

\textit{Notations}: In this paper, $a$ is a scalar, $\mathbf{a}$ is a column vector, $\mathbf{A}$ is a matrix, and $\mathbb{A}$ is a set. $\mathbf{A}^{\mathrm{T}}, \mathbf{A}^*$ denote the transpose and Hermitian of $\mathbf{A}$, respectively. $\mathrm{diag}\left(\cdot\right)$ denotes a diagonal matrix with given diagonal elements, and $\mathbf{0}$ denotes an all-zero matrix. $\left|\cdot\right|$ denotes the modulus of a complex number or the cardinality of a set and $\left|\left|\cdot\right|\right|_2$ denotes the $\ell_2$-norm of the vector. $\left(a, b\right)$ and $\left[a, b\right]$ are the open and closed interval on the real line, respectively. $\mathbb{R}$ and $\mathbb{C}$ denote the real and complex number sets, respectively, and $\mathbb{C}^{M\times N}$ denotes the space of $M\times N$ complex-valued matrices for natural number $M$ and $N$. $j\triangleq\sqrt{-1}$ is the imaginary unit, and $\mathcal{O}\left(\cdot\right)$ denotes the big-O notation. $\mathcal{E}\left(\cdot\right)$ and $\mathrm{tr}\left(\cdot\right)$ denotes the expectation and trace operator, respectively, and $\mathcal{U}\left(\cdot\right)$ denotes the uniform probability distribution with given interval.
					 \begin{center}
	\begin{table}[t] 
	\centering
		\caption{Summary of System Parameters}
		\begin{tabular}{|>{\centering } m{3cm} |>{\centering} m{4.5cm} |}
			\hline
			\textbf{Paramet{\tiny }er} & \textbf{Description} 
			\tabularnewline
			\hline
			\centering			$\mathbb{M}=\left\{1, \cdots, M_0\right\}$  & Index set of UAV-BSs   \tabularnewline \hline
			\centering			$\mathbb{U}_m$  & Set of users serviced by UAV-BS $m$   \tabularnewline \hline
			\centering			$\bm{\rho}_m\triangleq\left[\mathbf{w}_m^{\mathrm{T}}~h_m\right]^{\mathrm{T}}$  & 3D coordinate of UAV-BS $m$   \tabularnewline \hline
			\centering			$\left(M, N\right)$  & Number of antenna/RIS elements   \tabularnewline \hline
			\centering			$\lambda$  & Link wavelength   \tabularnewline \hline
			\centering			$G_{\mathrm{s}}$  & Antenna gain   \tabularnewline \hline
			\centering			$\left(d_{\mathrm{s}}, d_{\mathrm{RIS}}\right)$  & Antenna and RIS-element separation   \tabularnewline \hline
			\centering			$\left(\bar{d_{\mathrm{s}}}, \bar{d}\right)$  & $\lambda$-normalized\\antenna/RIS-element separation  \tabularnewline \hline
			\centering			$\bm{\rho}_{\mathrm{RIS}}\triangleq\left[\mathbf{q}^{\mathrm{T}}~H\right]^{\mathrm{T}}$  & 3D coordinate of aerial-RIS  \tabularnewline \hline
			\centering			$\beta_0$  & Reference path loss\\at a distance of 1 m  \tabularnewline \hline
			\centering			$\left(\beta_{\mathrm{s}}\left(\mathbf{q}\right), \beta\left(\mathbf{q}, \bm{\rho}_m\right)\right)$  & Path loss\\(source-RIS, RIS-UAV-BS $m$)   \tabularnewline \hline				
			\centering			$R_m$  & Backhaul rate of UAV-BS $m$   \tabularnewline \hline
			\centering			$C_m$  & Throughput of UAV-BS $m$   \tabularnewline \hline	
			\centering			$\sigma^2$  & Noise power  \tabularnewline \hline	
			\centering			\textcolor{red}{$N_{\mathrm{psd}}$}  & \textcolor{red}{Noise PSD ($=-174~\mathrm{dBm}$)}  \tabularnewline \hline	
			\centering			$\bm{\Theta}$  & RIS phase shift matrix   \tabularnewline \hline	
			\centering			$\theta_n$  & Phase shift of the $n$th element of RIS   \tabularnewline \hline	
			\centering			$\mathbf{v}_m$  & Unit-magnitude precoding vector intended to UAV-BS $m$   \tabularnewline \hline
			\centering			$s_m$  & Unit-power signal\\intended to UAV-BS $m$   \tabularnewline \hline
			\centering			$P_m$  & Source transmit power\\intended to UAV-BS $m$   \tabularnewline \hline
			\centering			$P_{\max}$  & Feasible threshold of the\\source transmit power  \tabularnewline \hline
			\centering			$\mathbf{x}$  & Total transmit signal of the source   \tabularnewline \hline
			\centering			$\bar{N}\left(\le N\right)$  & Number of activated RIS elements  \tabularnewline \hline
			\centering			$\left(\mathbf{H}\left(\mathbf{q}\right), \mathbf{h}^*\left(\mathbf{q},\cdot\right)\right)$  & Channel models\\(source-RIS, RIS-destination)  \tabularnewline \hline
			\centering			$\left(\mathbf{a}_{\mathrm{s}}\left(\cdot\right), \mathbf{a}_{\mathrm{RIS}}\left(\cdot\right)\right)$  & Array response (source, RIS)   \tabularnewline \hline				
			\centering			$\left(\phi_{\mathrm{t,s}}\left(\mathbf{q}\right), \phi_{\mathrm{t,RIS}}\left(\mathbf{q}, \cdot\right)\right)$  & AoD of the source-RIS and RIS-destination link   \tabularnewline \hline		
			\centering			$\phi_{\mathrm{r,RIS}}  \left(\mathbf{q}\right)$  & AoA of the source-RIS link   \tabularnewline \hline	
			\centering			$\gamma_m$  & Received SNR of UAV-BS $m$   \tabularnewline \hline
			\centering			$g\left(\cdot\right)$  & Passive beamforming gain by aerial-RIS   \tabularnewline \hline
			\centering			$\Delta\phi_{\mathrm{HPBW}}\left(\cdot\right)$  & Half-power beamwidth of $g$   \tabularnewline \hline
			\centering			$\Delta\phi_m\left(\bar{\bm{\rho}}\right)$ & sin-AoD deviation between $\bm{\bar{\rho}}$ and $\bm{\rho}_m$   \tabularnewline \hline	
		\end{tabular}
		\label{SysPar}
	\end{table}
\end{center}
	\section{System Model}
	\subsection{Aerial Backhaul Model}
The important notations used in this paper are summarized in Table~\ref{SysPar}. As illustrated in Fig.~\ref{fig_overview}, we consider the urban area $\mathcal{G}$ on the horizontal plane with the origin $\bm{\rho}_{\mathcal{G}}$ and containing $N_0$ users equipped with an omnidirectional antenna and serviced by stationary UAV-BSs $\mathbb{M}=\left\{1,\cdots,M_0\right\}$. Additionally, we assume that the UAV-BSs are equipped with a directional antenna where its azimuth and elevation half-power beamwidth (HPBW) are not equal~\cite{beamwidth}. The 3D coordinate of UAV-BS $m$ is given by $\bm{\rho}_m\triangleq\left[\mathbf{w}_m^{\mathrm{T}}~h_m\right]^{\mathrm{T}}$, which consists of the 2D location $\mathbf{w}_m$ and its height $h_m$. 

To avoid the inter-cell interference, we assume that each UAV-BS services a different non-overlapping and non-empty user subset. Here, the Ellipse Clustering~\cite{Noh} is applied for servicing the ground users, which determines $M_0$ and extensively lowers the total transmit power of UAV-BSs by determining their 3D locations. In this case, the throughput of UAV-BS $m$ is given by
\begin{equation}
\label{throughput}
C_m=\sum_{n\in\mathbb{U}_m} \frac{B_{\mathrm{f}}}{\left|\mathbb{U}_m\right|} \log_2 \left(1+ \frac{P_{\mathrm{f},n}}{\sigma_{\mathrm{f}}^2}\right),
\end{equation}
where $\mathbb{U}_m$ is the set of users serviced by UAV-BS $m$, $B_{\mathrm{f}}$ is the transmission bandwidth of the fronthaul link equally divided into $\left|\mathbb{U}_m\right|$ bands for each user in $\mathbb{U}_m$, $P_{\mathrm{f},n}$ is the received power of user $n\in\mathbb{U}_m$, and \textcolor{red}{$\sigma_{\mathrm{f},m}^2=\frac{B_{\mathrm{f}}}{\left|\mathbb{U}_m\right|} N_{\mathrm{psd}}$ is the noise power of the fronthaul link for UAV-BS $m$ with the power spectral density (PSD) of noise $N_{\mathrm{psd}}=-174~\mathrm{dBm}$.}

For a backhaul link via an aerial-RIS, we assume that the source is located at the origin and equipped with uniform-linear-array (ULA) with $M$ antennas. Each antenna has a gain of $G_{\mathrm{s}}$ and is separated by $d_{\mathrm{s}}$. We assume that the source-to-center distance $d_{\mathcal{G}}\triangleq\left|\left|\bm{\rho}_{\mathcal{G}}\right|\right|_2$ is sufficiently large and that the direct link from the source to the UAV-BS is blocked. The aerial-RIS also consists of ULA with $N$ reflecting elements without power amplification, separated by $d_{\mathrm{RIS}}\left(\in\left[\frac{\lambda}{10}, \frac{\lambda}{5}\right]\right)$~\cite{UAVRIS, LCSI}, where $\lambda$ is the carrier wavelength. Note that by apposing the parallel ULA-structured RIS by $N^\prime$ times, it is feasible to extend our scenario to aerial-RIS equipped with an $N\times N^{\prime}$ uniform-planar array (UPA) with expected performance gain of $N^\prime$ times, which can be regarded as the subset optimization of UPA-RIS. With its fixed altitude $H$ and by letting the first reflection element be the reference point, the 3D coordinate of the aerial-RIS is given by $\bm{\rho}_{\mathrm{RIS}}\triangleq\left[\mathbf{q}^{\mathrm{T}}~H\right]^{\mathrm{T}}$. Without loss of generality, we set the RIS array parallel to the x-axis. Moreover, due to the altitude of aerial-RIS, we assume that the backhaul link is dominated by an LoS component. Owing to the sufficiently small $d_{\mathrm{RIS}}$ compared to $H$ and $d_{\mathcal{G}}$, which determine the link distance, the backhaul link can be well approximated as uniform plane waves and we assume that the attenuation of the path can be considered as identical for all RIS-element pairs. Hence, the path loss of source-to-RIS $\beta_{\mathrm{s}}\left(\mathbf{q}\right)$ link and RIS-to-UAV-BS $m$ link $\beta\left(\mathbf{q},  \bm{\rho}_m\right)$ are given by~\cite{ITU525}
\begin{equation}
	\label{gain}
	\beta_{\mathrm{s}}\left(\mathbf{q}\right)=\frac{\beta_0}{\left|\left|\bm{\rho}_{\mathrm{RIS}}\right|\right|_2^2},~\beta\left(\mathbf{q},  \bm{\rho}_m\right)=\frac{\beta_0}{\left|\left|\bm{\rho}_{\mathrm{RIS}}-\bm{\rho}_m\right|\right|_2^2},
\end{equation}
where $\beta_0$ is the reference path loss at a link distance of 1~m given by $\beta_0=-20\log_{10}f-32.45$~[dB] for the link frequency $f$ measured in GHz. As we assumed ULA for both the source and the RIS, the channels between source-RIS $\mathbf{H}\left(\mathbf{q}\right)\in\mathbb{C}^{N\times M}$ and RIS-destination $\mathbf{h}^*\left(\mathbf{q},\cdot\right)\in\mathbb{C}^{1\times N}$ can be modeled by the angle-of-departure and arrival (AoD/AoA) of the communication links, thus containing every information about the azimuth and zenith angles of departure and arrival. Since the links are dominated by LoS path only, they are given~as
\begin{equation}
	\begin{split}
	\label{ULA}
	\begin{cases}
	\mathbf{H}\left(\mathbf{q}\right)\\
	=\sqrt{\beta_{\mathrm{s}} \left(\mathbf{q}\right)}e^{j\Phi_{\mathbf{H}}}e^{-j \frac{2\pi\left|\left|\bm{\rho}_{\mathrm{RIS}}\right|\right|_2}{\lambda}} \mathbf{a}_{\mathrm{RIS}} \left(\phi_{\mathrm{r,RIS}} \left(\mathbf{q}\right)\right) \mathbf{a}_{\mathrm{s}}^* \left(\phi_{\mathrm{t,s}} \left(\mathbf{q}\right)\right),\\
	\mathbf{h}^*\left(\mathbf{q},\bm{\rho}_m\right)\\
	=\sqrt{\beta \left(\mathbf{q}, \bm{\rho}_m\right)}e^{j\Phi_{\mathbf{h}}}e^{-j \frac{2\pi \left|\left|\bm{\rho}_{\mathrm{RIS}} - \bm{\rho}_m\right|\right|_2}{\lambda}} \mathbf{a}_{\mathrm{RIS}}^* \left(\phi_{\mathrm{t,RIS}} \left(\mathbf{q}, \bm{\rho}_m\right)\right),
\end{cases}
\end{split}
\end{equation}
where $\Phi_{\mathbf{H}}$ and $\Phi_{\mathbf{h}}$ are the random phase which are independent and generated with $\mathcal{U}\left[0, 2\pi\right)$, and $\mathbf{a}_{\mathrm{s}}\left(\cdot\right)\in\mathbb{C}^M$, $ \mathbf{a}_{\mathrm{RIS}}\left(\cdot\right)\in\mathbb{C}^N$ are the array response of the source and the aerial-RIS, respectively. They are given by
\begin{equation}
\label{ar}
\begin{split}
\begin{cases}
\mathbf{a}_{\mathrm{s}}\left(\cdot\right)=\left[1~e^{-j2\pi\bar{d}_{\mathrm{s}} \left(\sin\left(\cdot\right)\right)}~\cdots~e^{-j2\pi\left(M-1\right)\bar{d}_{\mathrm{s}} \left(\sin\left(\cdot\right)\right)}\right]^{\mathrm{T}},\\
\mathbf{a}_{\mathrm{RIS}}\left(\cdot\right)=\left[1~e^{-j2\pi\bar{d} \left(\sin\left(\cdot\right)\right)}~\cdots~e^{-j2\pi\left(N-1\right)\bar{d} \left(\sin\left(\cdot\right)\right)}\right]^{\mathrm{T}},
\end{cases}
\end{split}
\end{equation}
where $\bar{d}_{\mathrm{s}} \triangleq \frac{d_{\mathrm{s}}}{\lambda}$ and $\bar{d}\triangleq\frac{d_{\mathrm{RIS}}}{\lambda}$. Finally, $\phi_{\mathrm{t,s}}\left(\mathbf{q}\right)$, $\phi_{\mathrm{t,RIS}}\left(\mathbf{q}, \cdot\right)$ and $\phi_{\mathrm{r,RIS}} \left(\mathbf{q}\right)$ are the AoD of the source-RIS and RIS-destination link and the AoA of the source-RIS link, respectively. To make the targeted problem more tractable, we assume that the ground source perfectly knows all the RIS channels $\left\{\mathbf{h}^* \left(\mathbf{q}, \bm{\rho}_m\right)\right\}_{m=1}^{M_0}$ and $\mathbf{H}\left(\mathbf{q}\right)$, which can be acquired by the methods described in e.g.~\cite{CE1, CE2}.

By concatenating the channels in~(\ref{ULA}) and the phase shift of the RIS, the backhaul rate of UAV-BS $m$ is given by
\begin{equation}
	\label{rate}
	R_m=\frac{B_{\mathrm{b}}}{M_0} \log_2 \left(1+\underbrace{\frac{P_m G_{\mathrm{s}}\left|\mathbf{h}^*\left(\mathbf{q},\bm{\rho}_m\right) \bm{\Theta} \mathbf{H}\left(\mathbf{q}\right)\mathbf{v}_m\right|^2}{\sigma^2}}_{\triangleq\gamma_m}\right),
\end{equation}
where $B_{\mathrm{b}}$ is the transmission bandwidth of the backhaul link equally divided into $M_0$ bands for each UAV-BS, \textcolor{red}{$\sigma^2=\frac{B_{\mathrm{b}}}{M_0} N_{\mathrm{psd}}$ is the noise power of the backhaul link,} $\gamma_m$ is the received SNR of UAV-BS $m$, and $\mathbf{v}_m\in\mathbb{C}^M$ is a unit-magnitude precoding vector corresponding to unit-power signal $s_m$ intended to UAV-BS $m$ with source transmit power $P_m$. The total transmit signal $\mathbf{x}$ of the source is
\begin{equation}
\label{tx}
\mathbf{x}=\sum_{m\in \mathbb{M}} \mathbf{v}_m \sqrt{P_m G_{\mathrm{s}} }s_m.
\end{equation}
Finally, $\bm{\Theta}\triangleq\mathrm{diag}\left(\left\{\alpha_n e^{j\theta_n}\right\}_{n=1}^N\right)\in\mathbb{C}^{N\times N}$ is a diagonal phase shift matrix with the amplitude reflection coefficient $\alpha_n\in\left[0, 1\right]$ and phase shift $\theta_n \in \left[0,2\pi\right)$ of the $n$th element. Since we assume no power amplification of RIS, we can set $\alpha_n=1 \left(\forall n=1,\cdots, N\right)$. Note that the equivalent channel matrix $\mathbf{X}$ in~(\ref{rate}) that concatenates the the channels in~(\ref{ULA}) and the phase shift information is given by $\mathbf{X}\triangleq\mathbf{h}^*\left(\mathbf{q},\bm{\rho}_m\right) \bm{\Theta} \mathbf{H}\left(\mathbf{q}\right)$.
	\begin{figure}[t]
	\centering
	\begin{center}
		\includegraphics[width=0.96\columnwidth,keepaspectratio]%
		{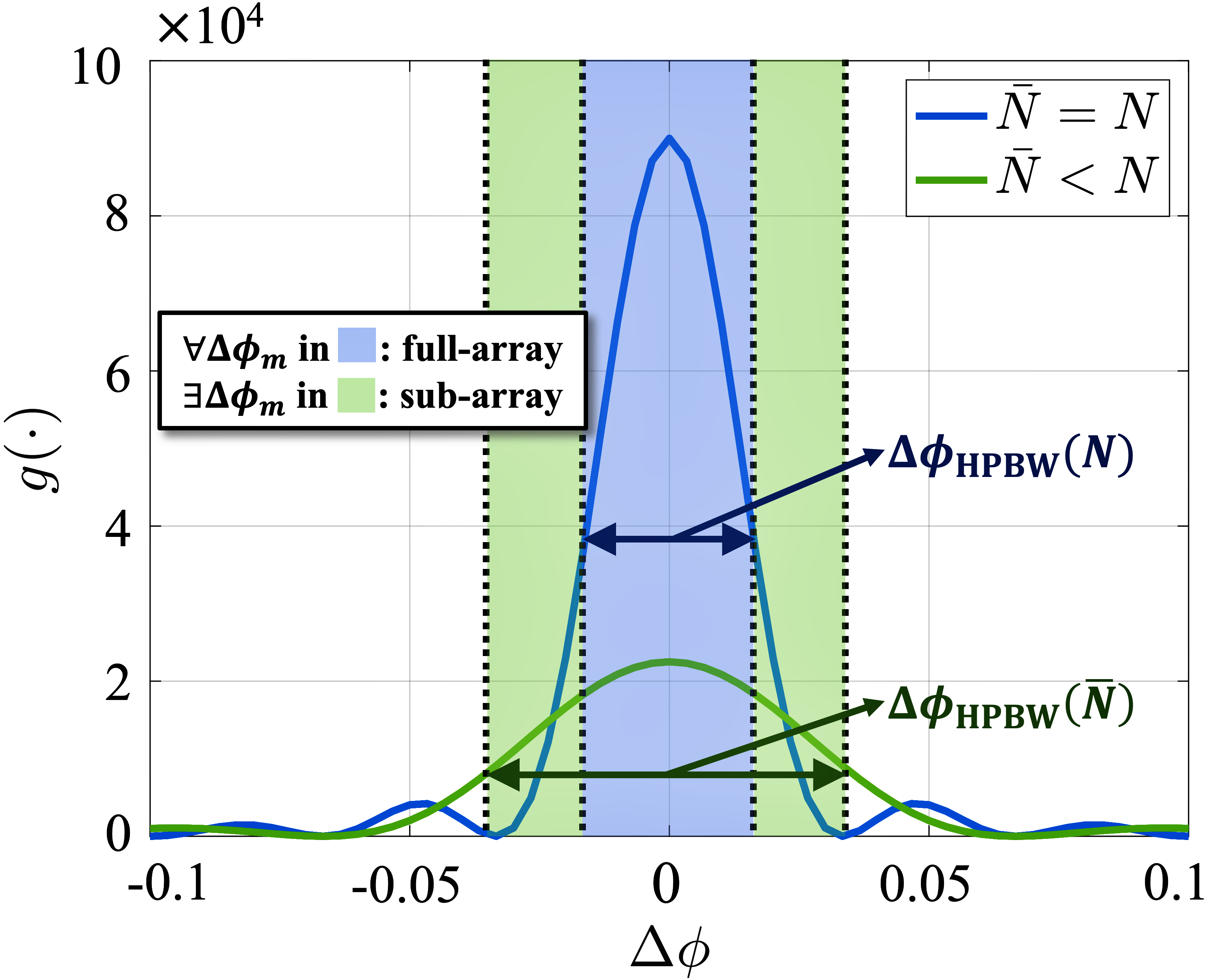}
		\caption{Passive beamforming gain $g$ and the region of full/sub-array structure. If the sin-AoD deviation exceeds the HPBW of the beamforming gain with full-array structure, we apply the sub-array structure to include the deviated point.}
		\label{figgain}
	\end{center}
\end{figure}

Here, to minimize the source transmit power $\sum_{m\in \mathbb{M}} P_m$ for high energy-efficiency~\cite{Noh, mozaenergy}, we first determine the precoding vector $\mathbf{v}_m$ of UAV-BS $m$ that maximizes the received SNR $\gamma_m\triangleq\frac{P_m G_{\mathrm{s}}\left|\mathbf{h}^*\left(\mathbf{q}, {\bm{\rho}}_m\right) \bm{\Theta} \mathbf{H}\left(\mathbf{q}\right)\mathbf{v}_m\right|^2}{\sigma^2}$ for fixed $P_m$, which leads to a smaller transmit power while achieving the same rate. Lemma~\ref{ThmV} presents the result.
	\begin{lemma}
	\label{ThmV}
	To maximize $\gamma_m$, the source should apply an MRT strategy. In other words, $\mathbf{v}_m$ is given by
	\begin{equation}
		\label{vm}
		\mathbf{v}_m=\mathbf{v}\triangleq\frac{\mathbf{a}_{\mathrm{s}}\left(\phi_{\mathrm{t,s}} \left(\mathbf{q}\right)\right)}{\left|\left|\mathbf{a}_{\mathrm{s}}\left(\phi_{\mathrm{t,s}} \left(\mathbf{q}\right)\right)\right|\right|_2}~\left(\forall m\in\mathbb{M}\right).
	\end{equation}
\end{lemma}
\begin{proof}
	By definition of $\mathbf{X}$
	, it is evident that
	\begin{equation}
		\label{SNR}
		\gamma_m=\frac{P_m G_{\mathrm{s}}}{\sigma^2} \mathbf{v}_m^* \mathbf{X}^* \mathbf{X} \mathbf{v}_m.
		\end{equation}
	By~(\ref{ULA}), $\mathbf{X}^* \mathbf{X}$ has a form of
	\begin{equation}
		\label{RR}
		\mathbf{X}^* \mathbf{X} =C_0  \mathbf{a}_{\mathrm{s}} \left(\phi_{\mathrm{t,s}} \left(\mathbf{q}\right)\right) \mathbf{a}_{\mathrm{s}}^* \left(\phi_{\mathrm{t,s}} \left(\mathbf{q}\right)\right),
		\end{equation}
	where $C_0$ is a positive constant. Equation~(\ref{RR}) implies that $\mathbf{X}^* \mathbf{X}$ is a rank-1 matrix with corresponding non-zero eigenpair:
	\begin{equation}
	\label{ep}
	\left(C_0 \left|\left|\mathbf{a}_{\mathrm{s}}\left(\phi_{\mathrm{t,s}} \left(\mathbf{q}\right)\right)\right|\right|_2^2, \frac{\mathbf{a}_{\mathrm{s}}\left(\phi_{\mathrm{t,s}} \left(\mathbf{q}\right)\right)}{\left|\left|\mathbf{a}_{\mathrm{s}}\left(\phi_{\mathrm{t,s}} \left(\mathbf{q}\right)\right)\right|\right|_2}\right).
	\end{equation}
	Hence, by using $\left|\left|\mathbf{v}_m\right|\right|_2=1$ and the Rayleigh-Ritz theorem, $\mathbf{v}_m$ is given by
	\begin{equation}
	\label{epep}
	\mathbf{v}_m=\frac{\mathbf{a}_{\mathrm{s}}\left(\phi_{\mathrm{t,s}} \left(\mathbf{q}\right)\right)}{\left|\left|\mathbf{a}_{\mathrm{s}}\left(\phi_{\mathrm{t,s}} \left(\mathbf{q}\right)\right)\right|\right|_2}
	\end{equation}
	to maximize $\mathbf{v}_m^* \mathbf{X}^* \mathbf{X} \mathbf{v}_m$ and the theorem follows.
\end{proof}
Note that the optimal transmission strategy is given by MRT since we have eliminated the interference by considering frequency division multiple access (FDMA) in~(\ref{rate}). By Lemma~\ref{ThmV}, MRT is adopted toward the aerial-RIS and $\gamma_m$ can be transformed into
\begin{equation}
\begin{split}
		\label{RSNR}
	\gamma_m=&\frac{P_m G_{\mathrm{s}}}{\sigma^2} \left|\mathbf{h}^*\left(\mathbf{q},\bm{\rho}_m\right) \bm{\Theta} \mathbf{H}\left(\mathbf{q}\right)\frac{\mathbf{a}_{\mathrm{s}}\left(\phi_{\mathrm{t,s}} \left(\mathbf{q}\right)\right)}{\left|\left|\mathbf{a}_{\mathrm{s}}\left(\phi_{\mathrm{t,s}} \left(\mathbf{q}\right)\right)\right|\right|_2}\right|^2 \\=&\bar{\gamma} \frac{\biggl|\sum_{n=1}^N e^{j\left(\theta_n +2\pi \left(n-1\right) \bar{d}\left(\sin \left(\phi_{\mathrm{t,RIS}}\left(\mathbf{q}, \bm{\rho}_m\right)\right)-\sin \left(\phi_{\mathrm{r,RIS}}  \left(\mathbf{q}\right)\right)\right)\right)} \biggr|^2}{\left|\left|\bm{\rho}_{\mathrm{RIS}}\right|\right|_2^2\left|\left|\bm{\rho}_{\mathrm{RIS}}-\bm{\rho}_m\right|\right|_2^2},
\end{split}
\end{equation}
where $\bar{\gamma} \triangleq\frac{P_m G_{\mathrm{s}}\beta_0^2 M}{\sigma^2 }$. From (\ref{RSNR}), we can notice that $\gamma_m$ is invariant of the antenna spacing $d_{\mathrm{s}}$.
\begin{figure*}[t]
	\begin{center}
		\includegraphics[width=1.4\columnwidth,keepaspectratio]%
		{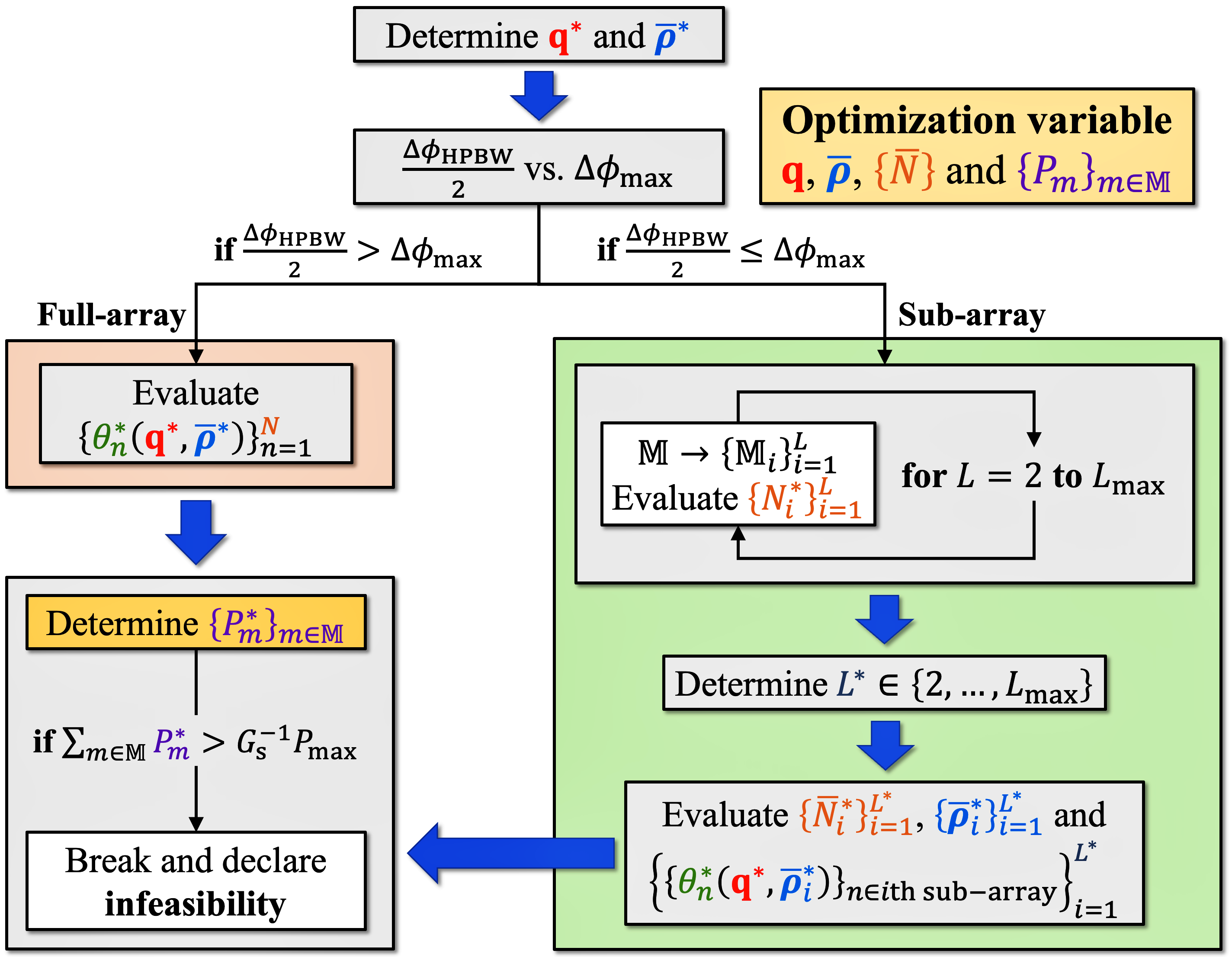}
		\caption{Block diagram of the proposed algorithm.}
		\label{block}
	\end{center}
\end{figure*}
\subsection{Aerial-RIS Beamforming Model}
To maximize $\gamma_m$ for a given $P_m$, we have to align the phase $\left\{\theta_n^*\right\}_{n=1}^N$ such that the reflected signals are coherently added to point $\bm{\rho}_m$:
\begin{equation}
\begin{split}
	\label{phasem}
	&\theta_n^*\left(\mathbf{q},\bm{\rho}_m\right)\\&=\bar\theta - 2\pi \left(n-1\right) \bar{d} \left(\sin \left(\phi_{\mathrm{t,RIS}}\left(\mathbf{q},\bm{\rho}_m\right)\right)-\sin\left(\phi_{\mathrm{r,RIS}}  \left(\mathbf{q}\right)\right)\right),
	\end{split}
\end{equation}
where $\bar\theta\in\left[0,2\pi\right)$ is an arbitrary phase shift. However, as there are multiple $M_0$ UAV-BSs to serve, $\left\{\theta_n^*\right\}_{n=1}^N$ differs for every $m\in\mathbb{M}$. Hence, we should determine the phase align point $\bm{\bar{\rho}}$ that leads to a Pareto-optimum for $\left\{\gamma_m\right\}_{m\in\mathbb{M}}$. In other words, for a given $\mathbf{q}$ and $\bm{\bar\rho}$, we set $\left\{\theta_n^*\right\}_{n=1}^N$ as
\begin{equation}
\begin{split}
	\label{phase}
	&\theta_n^*\left(\mathbf{q},\bm{\bar\rho}\right)\\&=\bar\theta-2\pi \left(n-1\right) \bar{d} \left(\sin \left(\phi_{\mathrm{t,RIS}}\left(\mathbf{q},\bm{\bar\rho}\right)\right)-\sin \left(\phi_{\mathrm{r,RIS}} \left(\mathbf{q}\right)\right)\right),
\end{split}
\end{equation}
which coherently adds the reflected signal to $\bm{\bar\rho}$. By substituting~(\ref{phase}) into~(\ref{RSNR}), $\gamma_m$ becomes
\begin{equation}
	\label{RSNR2}
	\gamma_m = \bar{\gamma}\frac{g\left(\Delta\phi_m\left(\bm{\bar\rho}\right)\right)}{ \left|\left|\bm{\rho}_{\mathrm{RIS}}-\bm{\rho}_m\right|\right|_2^2 \left|\left|\bm{\rho}_{\mathrm{RIS}}\right|\right|_2^2},
\end{equation}
where $g\left(\Delta\phi_m \left(\bm{\bar\rho}\right)\right)$ is the passive beamforming gain by aerial-RIS to $\bm{\rho}_m$ assuming the phases are aligned to $\bm{\bar\rho}$, which is obtained by manipulating $\biggl|\sum_{n=1}^N e^{j\left(\theta_n^*\left(\mathbf{q},\bm{\bar\rho}\right) +2\pi \left(n-1\right) \bar{d}\left(\sin \left(\phi_{\mathrm{t,RIS}}\left(\mathbf{q}, \bm{\rho}_m\right)\right)-\sin \left(\phi_{\mathrm{r,RIS}} \left(\mathbf{q}\right)\right)\right)\right)}\biggr|^2$ in~(\ref{RSNR}), with considering $\bar{N}\left(\le N\right)$ activated RIS elements, and given by
\begin{equation}
	\label{beamforming}
	g\left(\Delta\phi_m\left(\bm{\bar\rho}\right)\right)\triangleq\left|\frac{\sin\left(\pi \bar{N} \bar{d} \Delta \phi_m\left(\bm{\bar\rho}\right)\right)}{\sin \left(\pi \bar{d} \Delta \phi_m\left(\bm{\bar\rho}\right)\right)}\right|^2,
\end{equation}
where $\Delta\phi_m \left(\bm{\bar\rho}\right)$ is the sin-AoD deviation between $\bm{\bar\rho}$ and $\bm{\rho}_m$, that is,
\begin{equation}
\label{aodd}
\Delta\phi_m \left(\bm{\bar\rho}\right)\triangleq\sin \left(\phi_{\mathrm{t,RIS}} \left(\mathbf{q}, \bm{\rho}_m\right)\right)-\sin \left(\phi_{\mathrm{t,RIS}} \left(\mathbf{q}, \bm{\bar\rho}\right)\right).
\end{equation}
As illustrated in Fig.~\ref{figgain}, $g$ diminishes to 0 out of its HPBW, which is given by~\cite{HPBW}
\begin{equation}
\begin{cases}
\label{hpbwww}
\Delta\phi_{\mathrm{HPBW}}\left(\bar{N}\right)\approx\frac{0.8858}{\bar{N}\bar{d}},\\
\Delta\phi_{\mathrm{HPBW}} \left(N\right)\triangleq\Delta\phi_{\mathrm{HPBW}},
\end{cases}
\end{equation}
and the peak gain at $\Delta\phi_m (\cdot)=0$ is ${\bar{N}}^2$. Therefore, we have to fine tune $\bm{\bar\rho}$ and determine maximum $\bar{N}$ such that every UAV-BS locates in the HPBW of $g$, which leads to the maximization of $\left\{g\left(\Delta\phi_m \left(\bm{\bar\rho}\right)\right)\right\}_{m\in\mathbb{M}}$ in Section~III.B.
\subsection{Problem Formulation}
To reliably support the UAV-BSs, the backhaul rate should not be less than the throughput of the fronthaul link $\left\{C_m\right\}_{m\in\mathbb{M}}$: $R_m \ge C_m~\left(\forall m\in\mathbb{M}\right)$. By manipulating the equation using $R_m=\frac{B_{\mathrm{b}}}{M_0} \log_2 \left(1+\gamma_m \right)$, the transmit power $P_m$ must satisfy the following constraint:
\begin{equation}
	\label{Constraint}
	\begin{aligned}
		P_m\ge \left(2^{  \frac{M_0}{B_{\mathrm{b}}} C_m  }-1\right) \frac{  \sigma^2 \left|\left|\bm{\rho}_{\mathrm{RIS}}-\bm{\rho}_m \right|\right|_2^2 \left|\left|\bm{\rho}_{\mathrm{RIS}}\right|\right|_2^2   }{ G_{\mathrm{s}} \beta_0^2 M g\left(\Delta \phi_m\left(\bm{\bar\rho}\right)\right)   }~\left(\forall m\in\mathbb{M}\right).
	\end{aligned} 
\end{equation}
Moreover, by~(\ref{tx}) and Lemma~\ref{ThmV}, the maximum transmit constraint of the ground source is given by~\cite{RISEE}
\begin{equation}
\label{powerC}
\mathcal{E}\left[\left|\bm{\mathrm{x}}\right|^2\right]=\mathrm{tr}\left(G_{\mathrm{s}} \bm{\mathrm{PV^{\mathrm{H}}V}}\right)=G_{\mathrm{s}}\sum_{m\in\mathbb{M}} P_m\le P_{\mathrm{max}},
\end{equation}
wherein $\bm{\mathrm{V}}\triangleq\left[\bm{\mathrm{v}}_1 \cdots \bm{\mathrm{v}}_{M_0}\right]\in\mathbb{C}^{M\times M_0}$ derived by Lemma~\ref{ThmV}, $\bm{\mathrm{P}}\triangleq\mathrm{diag}\left(\left\{P_m\right\}_{m\in\mathbb{M}}\right)\in\mathbb{R}^{M_0\times M_0}$ and $P_{\mathrm{max}}$ is the feasible threshold of the source transmit power. Hence, the source power minimization problem can be formulated of
	\begin{equation}
		\label{object}
		\begin{aligned}
		& \underset{\mathbf{q}, \bm{\bar\rho}, \left\{\bar{N}\right\}, \left\{P_m\right\}_{m\in\mathbb{M}} }{\texttt{min}}~\sum\limits_{m\in\mathbb{M}} P_{m} \\
		&  \text{~~s.t.} ~P_m\ge \left(2^{   \frac{M_0}{B_{\mathrm{b}}} C_m    }-1\right) \frac{  \sigma^2 \left|\left|\bm{\rho}_{\mathrm{RIS}}-\bm{\rho}_m \right|\right|_2^2 \left|\left|\bm{\rho}_{\mathrm{RIS}}\right|\right|_2^2}{ G_{\mathrm{s}} \beta_0^2 M g\left(\Delta \phi_m\left(\bm{\bar\rho}\right)\right)}\\
		&~~~~~\left(\forall m\in\mathbb{M}\right), \sum_{m\in\mathbb{M}} P_m\le G_{\mathrm{s}}^{-1}P_{\mathrm{max}}.
		\end{aligned}	
\end{equation}
		Problem~(\ref{object}) is non-convex due to the product~$\left|\left|\bm{\rho}_{\mathrm{RIS}}-\bm{\rho}_m\right|\right|_2^2\left|\left|\bm{\rho}_{\mathrm{RIS}}\right|\right|_2^2$~and the beamforming function $g$ in~(\ref{beamforming}). We propose to solve it by minimizing the numerator and maximizing the denominator of the lower-bound of $P_m$ in~(\ref{object}) to minimize the sum. Fig.~\ref{block} shows a block diagram that illustrates the main steps for solving~(\ref{object}). In the following section, we discuss, in detail, each block of the proposed algorithm in Fig.~\ref{block}.
		\begin{remark}
		\label{r0}
		Since $C_m>0$, $\left|\left|\bm{\rho}_{\mathrm{RIS}}\right|\right|_2\ge H>0$ and $\left|\left|\bm{\rho}_{\mathrm{RIS}} - \bm{\rho}_m\right|\right|_2 >0$ ($\because \bm{\rho}_{\mathrm{RIS}}$ is extremely close to the origin. See Theorem~\ref{cubiceq} and Fig.~\ref{fig_cubic}), we can conclude that the right-hand side of~(\ref{Constraint}) is not zero, which therefore guarantees non-zero transmit power for every UAV-BS through aerial-RIS.
		\end{remark}
		\section{Proposed Aerial-RIS Setup Algorithm} 
		\subsection{Minimizing the Numerator: Determining $\mathbf{q}$}
		By letting $\mathbf{q}_m$ the 2D location of the aerial-RIS considering UAV-BS $m$ $\left(\bm{\rho}_{\mathrm{RIS}}=\left[\mathbf{q}_m^{\mathrm{T}}~H\right]^{\mathrm{T}}\right)$, 
		we can express the numerator minimization problem in regard to $\mathbf{q}_m$, which is provided~by
			\begin{equation}
			\label{num}
			\begin{aligned}
			&	 \underset{\mathbf{q}_m}{\texttt{min}}~\left|\left|\bm{\rho}_{\mathrm{RIS}}-\bm{\rho}_m \right|\right|_2^2 \left|\left|\bm{\rho}_{\mathrm{RIS}}\right|\right|_2^2 \\
			&	~~~~~~=\left(H^2 + \left|\left|\mathbf{q}_m\right|\right|_2^2 \right)\left(\left(H-h_m\right)^2 + \left|\left|\mathbf{q}_m-\mathbf{w}_m \right|\right|_2^2\right)\\
			& \text{~~s.t.} ~\left|\left|\mathbf{q}_m\right|\right|_2 \ll \delta \left|\left|\mathbf{w}_m\right|\right|_2,
			\end{aligned}	
		\end{equation}
	where $\delta$ is a sufficiently small positive constant. The constraint ``$\left|\left|\mathbf{q}_m\right|\right|_2\ll\delta \left|\left|\mathbf{w}_m\right|\right|_2$" is added to restrict $\mathbf{q}_m$ to be around the origin (source), as placing the RIS close to the source leads to the almost sure use of full-array RIS architecture $\left(\bar{N}=N\right)$, which maximizes the minimum SNR~\cite{UAVRIS} and therefore leads to smaller $\sum_{m\in\mathbb{M}}P_m$. Fortunately, we can find a practical solution for the problem, as stated in Theorem~\ref{cubiceq}.
	\begin{theorem}
	\label{cubiceq}
	The solution of Problem~(\ref{num}) is given by
	\begin{equation}
	\label{Ps}
	\mathbf{q}_m^* = \xi_m \mathbf{w}_m ,
	\end{equation}
	where
	\begin{equation}
	\begin{split}
	\label{xixixixi}
	&\xi_m=\frac{1}{2}+2\sqrt{-\frac{a}{3}}\cos \left( \frac{1}{3} \cos^{-1} \left( \frac{3b}{2a} \sqrt {-\frac{3}{a}} \right)-\frac{4}{3}\pi  \right),
	\end{split}
	\end{equation}
	$a$ and $b$ are given by
	\begin{equation}
	\label{xixixixixixi}
	a\triangleq\frac{1}{2}\left(\zeta_1^2 + \zeta_2^2 \right)- \frac{1}{4},~b\triangleq\frac{1}{4}\left(\zeta_2^2 - \zeta_1^2 \right),
		\end{equation}
		and
		\begin{equation}
		\label{ab}
	 \zeta_1\triangleq\frac{H}{\left|\left|\mathbf{w}_m\right|\right|_2} ,~\zeta_2 \triangleq \frac{\left|H-h_m\right|}{\left|\left|\mathbf{w}_m\right|\right|_2}.
	\end{equation}
		\end{theorem}
		\begin{proof}
See Appendix A.
	\end{proof}
 In Section~IV.B, we numerically verify that $\xi_m$ is greater than 0 but extremely close to 0 under the assumptions, which ascertains that $\mathbf{q}_m=\xi_m\mathbf{w}_m$ is a proper solution for our power-minimization procedure.
 \begin{remark}
 \label{r1}
 For $h_m=0$, the following holds:
 \begin{equation}
 \label{h0}
 \zeta_1 = \zeta_2 \triangleq\zeta \left(=\frac{H}{\left|\left|\mathbf{w}_m\right|\right|_2}\right)\rightarrow a=\frac{1}{2}\zeta^2-\frac{1}{4},~b=0.
 \end{equation}
 Therefore, $\xi_m$ becomes
 \begin{equation}
 \label{hzero}
 \xi_m=\frac{1}{2}-\sqrt{\frac{1}{4}-\zeta^2}.
 \end{equation}
 Equation~(\ref{hzero}) implies that $\xi_m$ goes to 0 for sufficiently large $\left|\left|\mathbf{w}_m\right|\right|_2$ $\left(\because \zeta\rightarrow0\right)$, which also corresponds to the derived~result.
 \end{remark}

After obtaining $\left\{\mathbf{q}_m^*\right\}_{m\in\mathbb{M}}$, we have to determine the location $\mathbf{q}^*$ from the result, which achieves a Pareto-optimum of~(\ref{num}) for every $m\in\mathbb{M}$. Under the prior knowledge of $\left\{\mathbf{q}_m^*\right\}_{m\in\mathbb{M}}$, we consider the following method of the global criterion that minimizes the sum of $\ell_2$-deviation and leads to the Pareto front~\cite{MVO}:
\begin{equation}
	\label{devmin}
	\begin{aligned}
		&	 \underset{\mathbf{q}}{\texttt{min}}~\sum_{m\in \mathbb{M}} \left|\left|\mathbf{q}_m^*-\mathbf{q} \right|\right|_2 .
	\end{aligned}	
\end{equation}
Problem~(\ref{devmin}) is called the Fermat-Torricelli problem, which is convex and can therefore be solved efficiently using the Weiszfeld's algorithm~\cite{FT}, which is proven to strictly converge to the optimal point of the problem~\cite{WF2}. Hence, we can determine $\mathbf{q}^*$ by~(\ref{devmin}), which leads to a suboptimal solution of the numerator minimization.
\begin{remark}
\label{r2}
To maintain small $\left|\left|\mathbf{q}^*\right|\right|_2$, although an outlier in $\left\{\mathbf{q}_m^*\right\}_{m\in\mathbb{M}}$ exists, we minimize the sum of the norm instead of the squared norm in~(\ref{devmin}), which additionally guarantees the robustness of the solution~\cite{MVO}.
\end{remark}
\begin{figure*}[t]
	\centering
	\subfloat[Full-array structure.]{
		\includegraphics[width=0.66\columnwidth]{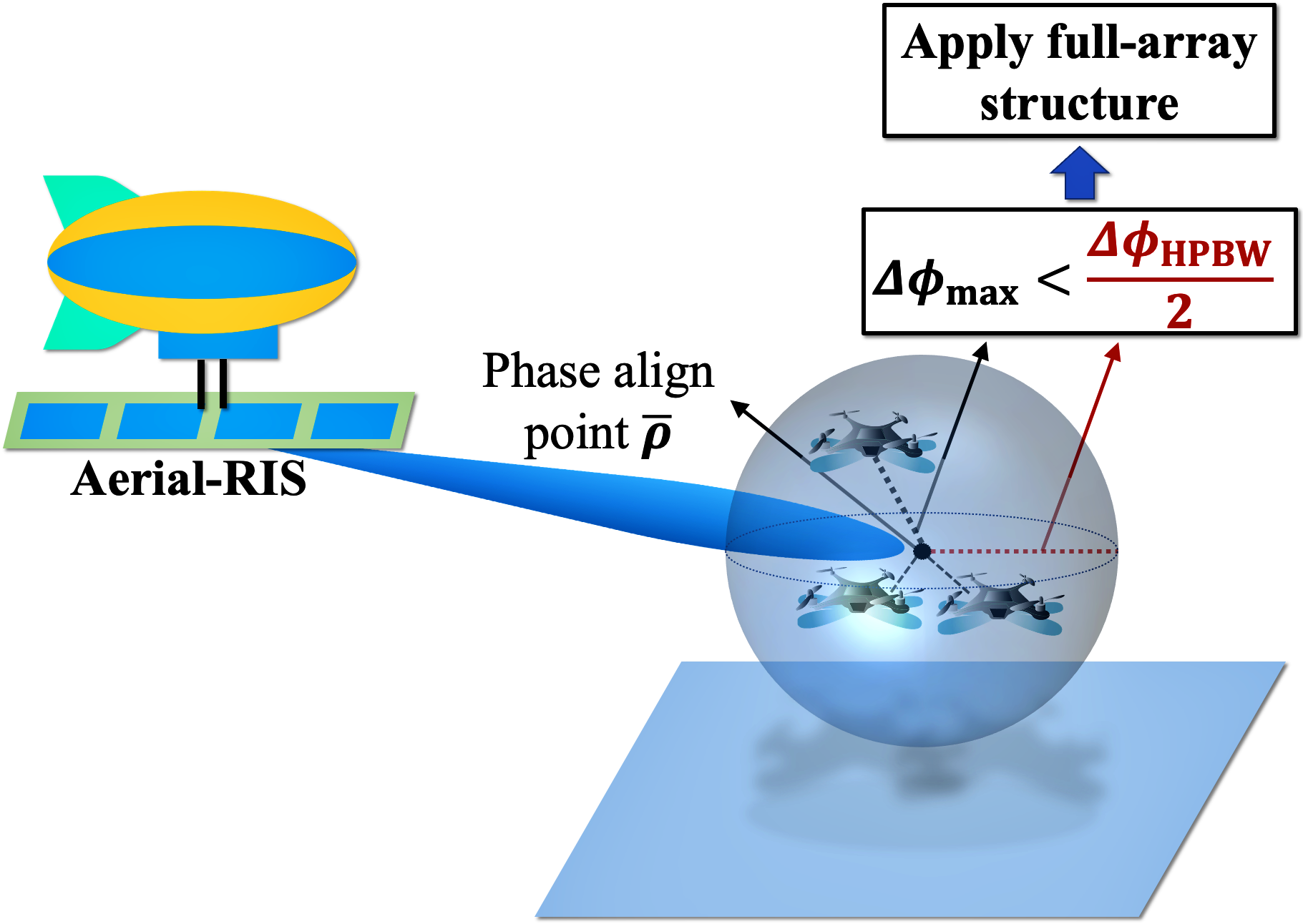}
		\label{figfull}
	}
	\subfloat[Sub-array structure with the proposed array-partition strategy.]{
		\includegraphics[width=1.34\columnwidth]{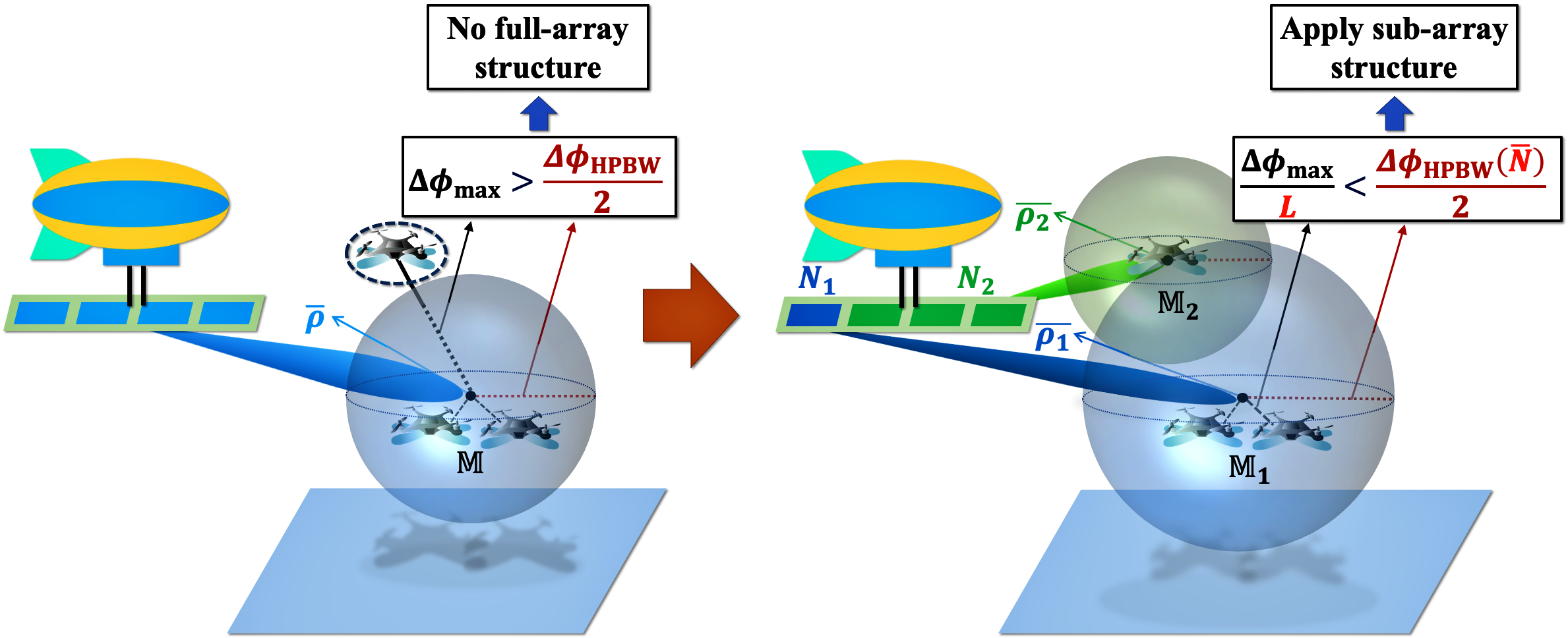}
		\label{figsub}
	}
	\caption{ 
		Array-structure selection of aerial-RIS with respect to HPBW of $g$ and maximum sin-AoD deviation. Due to the peak gain ${\bar{N}}^2$ placed in the denominator in~(\ref{Constraint}), the full-array structure leads to a further reduction in the transmit power.}	\label{fullsub}
\end{figure*}
		\subsection{Maximizing the Denominator: Determining $\{\bar{N}\}$, $\pmb{\bar\rho}$, and $\bm{\Theta}$}
		For a given $m\in\mathbb{M}$, it is known that $\Delta\phi_m\left(\cdot\right)=0$ leads to a maximum denominator, which is equivalent to $\bm{\bar\rho}=\bm{\rho}_m$. By the same logic, we have to determine $\bm{\bar\rho}$ that maximizes $\left\{g\left(\Delta\phi_m(\bm{\bar\rho})\right)\right\}_{m\in\mathbb{M}}$. As $\bm{\rho}_m$ and $\bm{\bar\rho}$ are sufficiently close, $\xi_m \ll1$ and $d_{\mathcal{G}}$ is sufficiently large, we can approximate the absolute value of the sin-AoD deviation by the first-order Taylor approximation with respect to $\bm{\rho}_m$:
		\begin{equation}
		\begin{split}
			\label{approxAoD}
			&\left|\Delta\phi_m\left(\bm{\bar\rho}\right)\right|\\
			&=\left|\sin \left(\phi_{\mathrm{t,RIS}} \left(\mathbf{q}^*, \bm{\rho}_m\right)\right)-\sin \left(\phi_{\mathrm{t,RIS}} \left(\mathbf{q}^*, \bm{\bar\rho}\right)\right)\right|\\
			&\approx\left|\cos\left(\phi_{\mathrm{t,RIS}} \left(\mathbf{q}^*, \bm{\rho}_m\right)\right) \left( \phi_{\mathrm{t,RIS}} \left(\mathbf{q}^*, \bm{\rho}_m\right)-\phi_{\mathrm{t,RIS}} \left(\mathbf{q}^*, \bm{\bar\rho}\right) \right)\right|\\
			&\approx\left|\cos\left(\phi_{\mathrm{t,RIS}} \left(\mathbf{q}^*, \bm{\rho}_m\right)\right)\right|\cdot\left|\left|\bm{\rho}_m-\bm{\bar\rho}\right|\right|_2 ~\left(\forall m\in\mathbb{M}\right).
		\end{split}
		\end{equation}
By~(\ref{approxAoD}), we aim to minimize the weighted $\ell_2$-deviation
\begin{equation}
\label{cosweight}
\left|\cos\left(\phi_{\mathrm{t,RIS}} \left(\mathbf{q}^*, \bm{\rho}_m\right)\right)\right|\cdot\left|\left|\bm{\rho}_m-\bm{\bar\rho}\right|\right|_2
\end{equation}
for every $m\in\mathbb{M}$. Similar to~(\ref{devmin}), to find the Pareto-optimum of $\left\{g\left(\Delta\phi_m\left(\bm{\bar\rho}\right)\right)\right\}_{m\in\mathbb{M}}$, we first find the phase align point $\bm{\bar\rho}^*$ assuming that every RIS array is activated for backhauling every UAV-BS in $\mathbb{M}$ (i.e., a full-array scenario), as illustrated in Fig.~\ref{figfull}. This is performed by solving the following weighted Fermat-Torricelli problem:
		\begin{equation}
		\label{devmin1.5}
		\begin{aligned}
			&	 \underset{\bm{\bar\rho}}{\texttt{min}}~\sum_{m\in \mathbb{M}}w_m \left|\left|\bm{\rho}_m-\bm{\bar\rho} \right|\right|_2\left(\approx\left|\Delta\phi_m\left(\bm{\bar\rho}\right)\right|\right),
		\end{aligned}	
	\end{equation}
		 where $w_m\triangleq\left|\cos\left(\phi_{\mathrm{t,RIS}} \left(\mathbf{q}^*, \bm{\rho}_m\right)\right)\right|~\left(\forall m\in\mathbb{M}\right)$. We can apply the weighted version of the Weiszfeld's algorithm given in Algorithm~\ref{wzf}, which also converges to the solution of~(\ref{devmin1.5}) at the Pareto front~\cite{WF2}.
		\begin{remark}
		\label{r4}
		 Note that the original Fermat-Torricelli problem can also be solved by letting $w_i=1~\left(\forall i=1, \cdots, m\right)$ in Algorithm~\ref{wzf}.
		 \end{remark}
	\begin{algorithm} [t]
	\caption{Weighted Version of the Weiszfeld's Algorithm}\label{wzf}
	\begin{algorithmic}[1]
		\Procedure{Minimizing the sum of $\ell_2$-deviation}{} \newline
		\textbf{Input:} $\left\{\mathbf{z}_i\right\}_{i=1}^m$ (no three points are colinear), $\epsilon_0 \ll 1$, $\left\{w_i\right\}_{i=1}^m \left(>0\right)$, $\mathbf{x}_0\notin \left\{\mathbf{z}_i\right\}_{i=1}^m$
		\State \multiline{Define $\mathbb{I}=\left\{1, \cdots, m\right\}$ and $W:\mathbb{R}^3 \rightarrow \mathbb{R}^3$ by
		\begin{equation}
		\label{wf}
		W\left(\mathbf{x}\right)=\frac{\sum_{i=1}^m w_i \frac{\mathbf{z}_i}{\left|\left|\mathbf{x}-\mathbf{z}_i\right|\right|_2}}{\sum_{i=1}^m w_i \frac{1}{\left|\left|\mathbf{x}-\mathbf{z}_i\right|\right|_2}}
		\end{equation}
		For continuity, define $W\left(\mathbf{z}_i\right)=\mathbf{z}_i \left(\forall i\in\mathbb{I}\right)$}
		\If {$\left|\left|\sum_{j=1, j\neq i}^m w_j \frac{\mathbf{z}_j - \mathbf{z}_i}{\left|\left|\mathbf{z}_j - \mathbf{z}_i\right|\right|_2} \right|\right|_2 >w_i \left(\forall i\in\mathbb{I}\right)$}
		\State Choose an initial point $\mathbf{x}_0\in\mathbb{R}^3 \backslash \left\{\mathbf{z}_i\right\}_{i=1}^m$
		\State Set $k=0$
		\While {$\left|\left|\mathbf{x}_{k+1}-\mathbf{x}_k\right|\right|_2 > \epsilon_0$}
		\State $\mathbf{x}_{k+1}=W\left(\mathbf{x}_k\right)$
		\State $k=k+1$
		\EndWhile
		\State $\mathbf{x}^*=\mathbf{x}_k$
		\Else~$\left(\exists i~\mathrm{s.t.} \left|\left|\sum_{j=1, j\neq i}^m w_j \frac{\mathbf{z}_j - \mathbf{z}_i}{\left|\left|\mathbf{z}_j - \mathbf{z}_i\right|\right|_2} \right|\right|_2 \le w_i\right)$
		\State $\mathbf{x}^*=\mathbf{z}_i$
		\EndIf
		\EndProcedure \\
		\textbf{Output:} $\mathbf{x}^*$
	\end{algorithmic}
\end{algorithm}
		 After deriving $\bm{\bar\rho}^*$, if there is a UAV-BS that locates out of the HPBW of $g$:
		 \begin{equation}
		 \label{star}
		 \Delta\phi_{\max}\triangleq\max_{m\in\mathbb{M}} \left|\Delta\phi_m\left(\bm{\bar\rho}^*\right)\right|>\frac{\Delta\phi_{\mathrm{HPBW}}}{2},
		 \end{equation}
		 we adopt the sub-array scenario $\left(\bar{N}<N\right)$, where the RIS array is divided into several sub-arrays and serve the UAV-BSs in each subset of $\mathbb{M}$, instead of full-array and vice versa. It is because the deviated UAV-BS has almost zero passive beamforming gain and significantly increases the corresponding transmit power (see~(\ref{Constraint}) and Fig.~\ref{figgain}).
		 \subsubsection{Full-array scenario $\left(\bar{N}=N\right)$}
		 For the full-array case, we use the phase align point $\bm{\bar\rho}^*$ derived by~(\ref{devmin1.5}) and the RIS phase $\left\{\theta_n^*\left(\mathbf{q}^*,\bm{\bar\rho}^*\right)\right\}_{n=1}^N$ is derived by~(\ref{phase}).
		\subsubsection{Sub-array scenario $\left(\bar{N}<N\right)$}
		 For the sub-array case, as illustrated in Fig.~\ref{figsub}, we divide the RIS array into $L$ partitions with $\left\{N_i\right\}_{i=1}^L$ elements for given $L$ and $\mathbb{M}$ into $\left\{\mathbb{M}_i\right\}_{i=1}^L$ corresponding to each $N_i$, where $\mathbb{M}_i$ is the set of UAV-BS $m_0$ that corresponds to one of the $L$ equal partitions of the feasible sin-AoD deviation range of $\mathbb{M}$:
	 \begin{equation}
	 	\label{division}
		\Delta\phi_{m_0} \left(\bm{\bar\rho}^*\right)\in \frac{\Delta\phi_{\max}}{L}\left(2\left(i-1\right)-L, 2i-L\right]\triangleq \mathbb{I}_i.	
	 \end{equation}
Next, we have to determine $\left\{N_i\right\}_{i=1}^L$ that leads to minimum transmit power. Assuming that $\bm{\rho}_{\mathrm{RIS}}$ is determined in Section~III.A, and every UAV-BS locates in the HPBW so that $g\approx\bar{N}^2$ holds for every sub-array, the lower-bound of $P_m$ in~(\ref{object}) can be approximated by 
\begin{equation}
\label{lblb}
\left(2^{   \frac{M_0}{B_{\mathrm{b}}} C_m  }-1\right) \frac{  \sigma^2 \left|\left|\bm{\rho}_{\mathrm{RIS}}-\bm{\rho}_m \right|\right|_2^2 \left|\left|\bm{\rho}_{\mathrm{RIS}}\right|\right|_2^2    }{ G_{\mathrm{s}} \beta_0^2 M g\left(\Delta \phi_m\left(\bm{\bar\rho}\right)\right)   }
\approx\frac{A_m}{{\bar{N}}^2}, 
\end{equation}
where $A_m$ is a positive constant. Moreover, we should prevent the ``HPBW-outlier" in every sub-array that satisfies the~following:
\begin{equation}
\label{outlier}
	 \left|\Delta\phi_m\left(\bm{\bar\rho}_i^*\right)\right|>\frac{\Delta\phi_{\mathrm{HPBW}} \left(N_i\right)}{2} ~\left(m\in\mathbb{M}_i\right),
\end{equation}
where $\bm{\bar\rho}_i^*$ is the phase align point of $i$th sub-array which will be determined by maximizing $\left\{g\left(\Delta\phi_m \left(\bm{\bar\rho}_i^*\right)\right)\right\}_{m\in\mathbb{M}_i}$ and described at the end of the section. By~(\ref{division}), the maximum sin-AoD deviation $\max_{m\in\mathbb{M}_i}\left|\Delta\phi_{m}\left(\bm{\bar\rho}_i^*\right)\right|$ of $\mathbb{M}_i$ is reduced by $L$ times~\cite{UAVRIS}, which is $\left|\mathbb{I}_i\right|= \frac{\Delta\phi_{\max}}{L}$. Therefore, to prevent the outlier, $N_i$ should satisfy the following: 
\begin{equation}
\begin{split}
\label{constN}
&  \frac{\Delta\phi_{\max}}{L} \le \frac{\Delta\phi_{\mathrm{HPBW}} \left(N_i\right)}{2}\\
  &~ \rightarrow N_i \le L \frac{\Delta\phi_{\mathrm{HPBW}}/2}{\Delta\phi_{\max}} N\left(\triangleq kN\right).
\end{split}
\end{equation}
Hence, to minimize the sum of the lower-bound of $P_m$ by determining $\left\{N_i\right\}_{i=1}^L$, we consider the following problem.
	\begin{equation}
		\label{water}
		\begin{aligned}
		& \underset{\left\{N_i\right\}_{i=1}^L }{\texttt{min}}~\sum_{i=1}^L\sum_{m\in\mathbb{M}_i} \frac{A_m}{N_{i}^2} \\
		&  \text{~~s.t.}  ~\sum_{\ell=1}^L N_\ell =N,~0\le N_i \le kN~\left(i=1,\cdots,L\right).
		\end{aligned}	
\end{equation}
For fixed $L$, Problem~(\ref{water}) is convex and therefore can be directly solved~\cite{boyd}. The result is given by Theorem~\ref{KKT}.
 			\begin{algorithm} [t]
	\caption{Proposed Aerial-RIS Setup Algorithm}\label{euclid}
	\begin{algorithmic}[1]
	\Procedure{Source power minimization}{}
		\State Find $\mathbf{q}_m^* = \xi_m\mathbf{w}_m$ by~(\ref{solcubic}) and determine $\mathbf{q}^*$ by~(\ref{devmin})
		\State Determine $\bm{\bar\rho}^*$ by~(\ref{devmin1.5}) 
		\State Compare $\frac{\Delta\phi_{\mathrm{HPBW}}}{2}$ and $\Delta\phi_{\max}\triangleq\max_{m\in\mathbb{M}} |\Delta\phi_m(\pmb{\bar\rho}^*)|$
		\If {$\frac{\Delta\phi_{\mathrm{HPBW}}}{2} > \Delta\phi_{\max}$}
				\State Apply full-array structure $\left(\bar{N}=N\right)$
				\State Find the RIS phase $\left\{\theta_n^*\left(\mathbf{q}^*,\bm{\bar\rho}^*\right)\right\}_{n=1}^N$ by~(\ref{phase})
				\Else ~(sub-array structure $\left(\bar{N}<N\right)$)
				\For {$L=2$ to $L_{\max}$}
				\State \multiline{Divide $\mathbb{M}$ into $\left\{\mathbb{M}_i\right\}_{i=1}^L$ based on~(\ref{division})}
				\State Divide the RIS array into $L$ sub-arrays by~(\ref{optimalssss})
		\EndFor
		\State \multiline{Determine~$L^*\in\left\{2, \cdots, L_{\max}\right\}$ by~(\ref{1dsearch})}
		\State Round the solution $\left\{N_i^*\right\}_{i=1}^{L^*} \rightarrow \left\{\bar{N}_i^*\right\}_{i=1}^{L^*}$
		\State Evaluate $\left\{\bm{\bar\rho}_i^*\right\}_{i=1}^{L^*}$ by~(\ref{devmin2})
		\State \multiline{Find $\left\{\left\{\theta_n^*\left(\mathbf{q}^*,\bm{\bar\rho}_i^*\right)\right\}_{n\in i\textrm{th~sub-array}}\right\}_{i=1}^{L^*}$ by~(\ref{phase})}
		\EndIf
		\State Determine $\left\{P_m^*\right\}_{m\in\mathbb{M}}$ by~(\ref{optP})
		\If {$\sum_{m\in\mathbb{M}} P_m^* > G_{\mathrm{s}}^{-1}P_{\mathrm{max}}$}
		\State Break and declare infeasibility
		\EndIf
		\EndProcedure
	\end{algorithmic}
\end{algorithm}
\begin{theorem}
\label{KKT}
The solution of Problem~(\ref{water}) is given by
\begin{equation}
\label{optimalssss}
 N_i^*=\min\left(kN, \sqrt[3]{\frac{\sum_{m\in\mathbb{M}_i} 2A_m}{\mu}}\right) \left(i=1,\cdots,L\right),
 \end{equation}
where $\mu$ is chosen such that $\sum_{\ell=1}^L N_\ell^* = N$ is met.
 \end{theorem}
 \begin{proof}
 See Appendix B.
  \end{proof}
  Intuitively, we set initial $N_i^*$ by $\sqrt[3]{\frac{\sum_{m\in\mathbb{M}_i} 2A_m}{\mu}}$, but if the value exceeds the threshold $kN$, we automatically restrict $N_i^*$ by $N_i^*=kN$. We repeat the procedure for $L\in\left\{2, \cdots, L_{\max}\right\}$ for positive integer $L_{\max}\left(\ge2\right)$, and find $L^*$ that leads to minimum objective function in~(\ref{water}) via one-dimensional search over $\left\{2, \cdots, L_{\max}\right\}$, that is,
\begin{equation}
\label{1dsearch}
L^*\triangleq \argmin_{L\in\left\{2, \cdots, L_{\max}\right\}} \sum_{i=1}^L \sum_{m\in\mathbb{M}_i} \frac{A_m}{N_i^{*2}}.
\end{equation}
Finally, since the number of RIS elements is an integer, we can round the result to give the integer solution $\left\{\bar{N}_i^*\right\}_{i=1}^{L^*}$. The proposed array-partition is given in Fig.~\ref{figkkt}. The proposed partition can be considered as the ``reversed" version of the waterfilling algorithm~\cite{CBF, CBF22}, where the role between the vessel and the water is changed.

For phase align point $\left\{\bm{\bar\rho}_i^*\right\}_{i=1}^{L^*}$ of sub-arrays corresponding to $\left\{\mathbb{M}_i\right\}_{i=1}^{L^*}$, by following~(\ref{approxAoD}) and~(\ref{devmin1.5}), we have to solve the following weighted Fermat-Torricelli problem for every $i\in\left\{1,\cdots,L^*\right\}$.
	\begin{equation}
		\label{devmin2}
		\begin{aligned}
			&	 \underset{\bm{\bar\rho}_i}{\texttt{min}}~\sum_{m\in \mathbb{M}_i} w_m \left|\left|\bm{\rho}_{m}-\bm{\bar\rho}_i \right|\right|_2~\left(i=1,\cdots, L^*\right).
		\end{aligned}	
	\end{equation}
Note that $w_m=\left|\cos\left(\phi_{\mathrm{t,RIS}} \left(\mathbf{q}^*, \bm{\rho}_m\right)\right)\right|$ is independent of $\bm{\bar\rho}$. By~(\ref{devmin2}), we can find the suboptimal solution $\left\{\bm{\bar\rho}_i^*\right\}_{i=1}^{L^*}$ for each sub-array and their corresponding RIS phases $\left\{\left\{\theta_n^*\left(\mathbf{q}^*,\bm{\bar\rho}_i^*\right)\right\}_{n\in i\textrm{th~sub-array}}\right\}_{i=1}^{L^*}$ by~(\ref{phase}).
\begin{figure}[t]
	\begin{center}
		\includegraphics[width=0.96\columnwidth,keepaspectratio]%
		{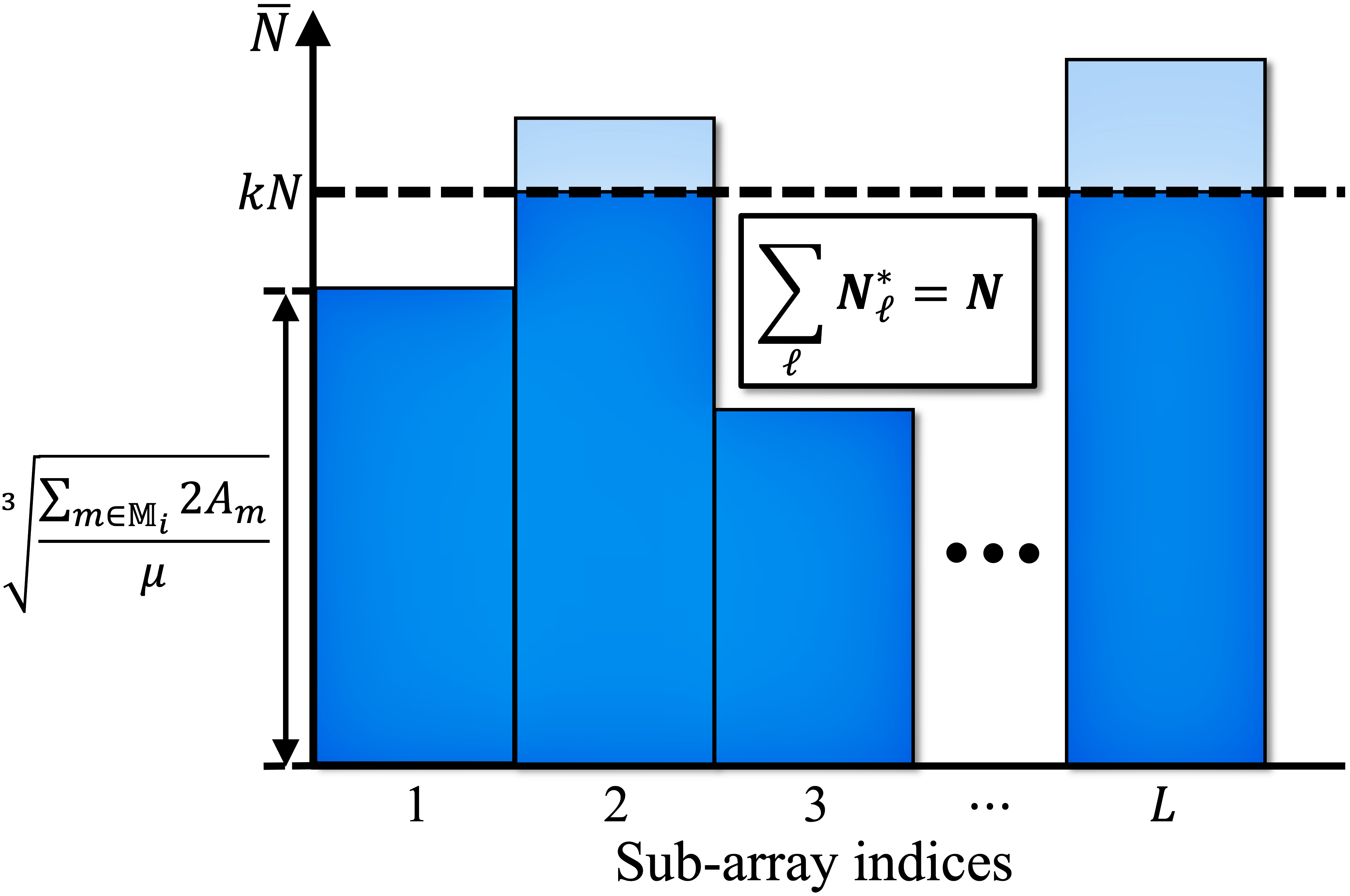}
		\caption{RIS array-partition strategy for sub-array scenario. We restrict the length of the sub-array by $kN$, and hence the lighter region of the graph is ignored in the final result $\left\{N_i^*\right\}_{i=1}^L$.}
		\label{figkkt}
	\end{center}
\end{figure}
\begin{table*}[t]
\centering
\caption{Complexity of the Proposed Algorithm}
\label{CT}
\begin{tabular}{|c|c|c|c|c|c|} 
\hline
\multicolumn{2}{|c|}{\textbf{Section III.A}}           & \multicolumn{2}{c|}{\textbf{Section III.B}}               & \multicolumn{2}{c|}{\textbf{Section III.C}}        \\ 
\hline
\textbf{Parameter}   & \textbf{Complexity}  & \textbf{Parameter}   & \textbf{Complexity}               & \textbf{Parameter}   & \textbf{Complexity}          \\ 
\hline
\multirow{3}{*}{$\left\{\xi_m\right\}_{m\in\mathbb{M}}$ } & \multirow{3}{*}{$\mathcal{O}(M_0)$}  & $\bm{\bar\rho}^*$   & $\le \mathcal{O}\left(I_{\mathbb{M}}M_0+M_0^2\right)$ & \multirow{6}{*}{$\left\{P_m^*\right\}_{m\in\mathbb{M}}$} & \multirow{6}{*}{$\mathcal{O}\left(M_0\right)$}  \\ 
\cline{3-4}
          &                 & $\left\{N_i^*\right\}_{i=1}^L,~\sum_{i=1}^L \sum_{m\in\mathbb{M}} \frac{A_m}{N_i^{*2}}$ & $\mathcal{O}\left(M_0\right)$                   &                    &                      \\ 
\cline{3-4}
             &  & $L^*
             $                  & $\mathcal{O}\left(L_{\max}M_0\right)$   &     &    \\ 
\cline{1-4}
\multirow{3}{*}{$\mathbf{q}^*$}                & \multirow{3}{*}{$\le \mathcal{O}\left(I_{\mathbb{M}} M_0 \right)$~} & $\left\{N_i^*\right\}_{i=1}^{L^*} \rightarrow \left\{\bar{N}_i^*\right\}_{i=1}^{L^*}$   & $\mathcal{O}\left(L^*\right)$ &       &   \\ 
\cline{3-4}
              &   &  $\left\{\bm{\bar\rho}_i^*\right\}_{i=1}^{L^*}$       & $\le \mathcal{O}\left(I_{L^*} M_0+M_0^2\right)$     &        &                      \\ 
\cline{3-4}
   &    & $\bm{\Theta}\triangleq\mathrm{diag}\left(\left\{e^{j\theta_n}\right\}_{n=1}^N\right) \left(\because\forall\alpha_n=1\right)$                   & $\mathcal{O}\left(N\right)$          && \\
\hline
\end{tabular}
\end{table*}
	 \begin{center}
	\begin{table}[t] 
	\centering
		\caption{Simulation Parameters}
		\begin{tabular}{|>{\centering } m{1.6cm} |>{\centering} m{3.9cm} |>{\centering} m{1.8cm} | }
			\hline
			\textbf{Paramet{\tiny }er} & \textbf{Description} & \textbf{Value}
			\tabularnewline
			\hline
			\centering			$B_{\mathrm{b}}$  & Bandwidth of the backhaul link\\(unless referred) & 50 (MHz)  \tabularnewline \hline
			\centering			$M_0$  & Number of UAV-BSs & $\gtrapprox4$~\cite{Noh}  \tabularnewline \hline
			\centering			\textcolor{red}{$\beta_0$}  & \textcolor{red}{Reference path loss\\for sub-6~GHz backhaul} & \textcolor{red}{-43.32 (dB)}  \tabularnewline \hline
			\centering			$\mathcal{G}$  & Targeted urban region& $500\times 500$ (m)  \tabularnewline \hline
			\centering			$\bm{\rho}_{\mathcal{G}}$  & Center of $\mathcal{G}$ (unless referred) &$[1000 ~0]^{\mathrm{T}}$ (m)  \tabularnewline \hline
			\centering                   $\delta$ & Upper-bound of restricting $\mathbf{q}_m^*$ & $10^{-1}$ \tabularnewline \hline
			\centering			$P_{\max}$  & Feasible threshold of\\source transmit power& 30 (dBW) \tabularnewline \hline			
			\centering			$G_{\max}$  & Maximum directional gain & 8 (dB) \tabularnewline \hline
			\centering			$\left(\mathrm{SLA_v}, A_{\max}\right)$  & Vertical side-lobe attenuation and maximum attenuation & 30 (dB)~\cite{NR} \tabularnewline \hline	
			\centering		$H$ & Height of aerial-RIS\\(unless referred) & 150 (m) \tabularnewline \hline
			\centering		$N$ & Number of RIS elements\\(unless referred) & 300 \tabularnewline \hline
			\centering		$M$ & Number of source antennas & 16 \tabularnewline \hline
			\centering		$L_{\max}$ & Upper-bound of\\one-dimensional search & 5 \tabularnewline \hline
			\centering      $\left(d_{\mathrm{s}}, d_{\mathrm{RIS}}\right)$ & Source antenna and\\RIS element separations & $\left(\frac{\lambda}{2}, \frac{\lambda}{10}\right)$~\cite{meta} \tabularnewline \hline
		\end{tabular}
		\label{SimPar}
	\end{table}
\end{center}
\subsection{Determining $\left\{P_m\right\}_{m\in\mathbb{M}}$}
Since we determine the placement of aerial-RIS and the phase align point/points in Sections~III.A and B, we apply the result to the constraints in~(\ref{object}). In turn, Problem~(\ref{object}) becomes linear programming with respect to $\left\{P_m\right\}_{m\in\mathbb{M}}$. The feasible solution satisfies the equality of the lower-bound of the problem for every $m\in\mathbb{M}$. It is given by
\begin{equation}
	\label{optP}
	P_m ^*= \left(2^{   \frac{M_0}{B_{\mathrm{b}}} C_m    }-1\right) \frac{  \sigma^2 \left|\left|\bm{\rho}_{\mathrm{RIS}}^*-\bm{\rho}_m \right|\right|_2^2 \left|\left|\bm{\rho}_{\mathrm{RIS}}^*\right|\right|_2^2    }{ G_{\mathrm{s}} \beta_0^2 M g\left(\Delta \phi_m^*\right)   }~\left(\forall m\in\mathbb{M}\right),
\end{equation}
where $\bm{\rho}_{\mathrm{RIS}}^*=\left[\mathbf{q}^{* \mathrm{T}}~H\right]^{\mathrm{T}}$ and $\Delta\phi_m^*$ is defined by
\begin{equation}
\label{finalphi}
\Delta\phi_m^* \triangleq
\begin{cases}
\Delta\phi_m\left(\bm{\bar\rho}^*\right) & \left(\text{full-array scenario}\right) \\
\Delta\phi_m \left(\bm{\bar\rho}_i^*\right) & \left(\text{sub-array scenario}, m\in\mathbb{M}_i\right).
\end{cases}
\end{equation}
The complete algorithm is summarized in Algorithm~\ref{euclid}.
\begin{remark}
\label{ri}
Although the power is guaranteed to be positive in \textit{Remark}~\ref{r0}, no strict feasibility claim can be made since we cannot always guarantee $\sum_{m\in\mathbb{M}} P_m^* \le G_{\mathrm{s}}^{-1}P_{\mathrm{max}}$ by~(\ref{optP}). In Section IV, the practical transmit power given by the proposed algorithm is way less (below 40$\sim$45 dBm) than the feasible threshold (the feasibility rate is 100$\%$ for our algorithm), and this ensures the almost-sure feasibility. Therefore, we can conclude that the consideration of the infeasibility scenario is~negligible.
\end{remark}
		\subsection{Complexity Analysis}
		The computational complexity of the proposed algorithm consists of three parts. In Section~III.A, we first choose $\left\{\xi_m\right\}_{m\in\mathbb{M}}$ by~(\ref{solcubic}), which requires a complexity of $\mathcal{O}\left(M_0\right)$. Next, we determine $\mathbf{q}^*$ by~(\ref{devmin}) and Algorithm~\ref{wzf} with complexity upper-bounded by $\mathcal{O}\left(I_{\mathbb{M}} M_0\right)$, where $I_{\left(\cdot\right)}$ is the number of iterations from procedure 4$\sim$7 in the Weiszfeld's algorithm for given set.
		In Section~III.B, we first find $\bm{\bar\rho}^*$ by~(\ref{devmin1.5}) with complexity upper-bounded by $\mathcal{O}\left(I_{\mathbb{M}} M_0 + M_0^2 \right)$. For full-array scenarios, we align the RIS phases by~(\ref{phase}), which requires a complexity of $\mathcal{O}\left(N\right)$. For sub-array scenarios, we find $\left\{N_i^*\right\}_{i=1}^L$ for given $L$ by~(\ref{optimalssss}) with a complexity of $\sum_{i=1}^L \mathcal{O}\left(\left|\mathbb{M}_i\right|\right)=\mathcal{O}\left(M_0\right)$, and compute the objective function $\sum_{i=1}^L \sum_{m\in\mathbb{M}_i} \frac{A_m}{N_i^{*2}}$, which also has a complexity of $\sum_{i=1}^L \mathcal{O}\left(\left|\mathbb{M}_i\right|\right)=\mathcal{O}\left(M_0\right)$. We repeat the procedure for $L=2, \cdots, L_{\max}$ and search $L^*$ by~(\ref{1dsearch}), which has a complexity of $\mathcal{O}\left(L_{\max} M_0\right)$. We finally round the solution $\left\{N_i^*\right\}_{i=1}^{L^*}$ with a complexity of $\mathcal{O}\left(L^*\right)$. Next, we find the phase align point $\left\{\bm{\bar\rho}_i^*\right\}_{i=1}^{L^*}$ by~(\ref{devmin2}) with complexity upper-bounded by $\sum_{i=1}^{L^*} \mathcal{O}\left(I_{\mathbb{M}_i}\left|\mathbb{M}_i \right| + \left|\mathbb{M}_i\right|^2\right) \le \mathcal{O}\left(I_{L^*}M_0+M_0^2\right)$, where $I_{L^*}\triangleq\max_{i\in\left\{1,\cdots,L^*\right\}} I_{\mathbb{M}_i}$. Then, we align the RIS phases by~(\ref{phase}), which has a complexity of $\mathcal{O}\left(N\right)$. Finally, in Section~III.C, we determine $\left\{P_m^*\right\}_{m\in\mathbb{M}}$ by~(\ref{optP}), which requires a complexity of $\mathcal{O}\left(M_0\right)$.
		
		Hence, the total complexity is upper-bounded by
		\begin{equation}
		\label{FC}
		\begin{split}
		&\mathcal{O}\left(\left(I_{\mathbb{M}} + L_{\max} + I_{L^*} + M_0 \right)M_0 +L^*+N\right)\\&\approx \mathcal{O}\left(\left(I_{\mathbb{M}}  + I_{L^*} +M_0 \right)M_0 +N\right)~\left(\because I_{\left(\cdot\right)} > L_{\max}\ge L^*\right),
		\end{split}
		\end{equation}
		which does not exceed the quadratic order. We can therefore conclude that our proposed algorithm is not only energy-efficient but also complexity-efficient. The complexity of each step of the proposed algorithm is summarized in Table~\ref{CT}.
	\begin{figure}[t]
	\begin{center}
		\includegraphics[width=0.97\columnwidth,keepaspectratio]%
		{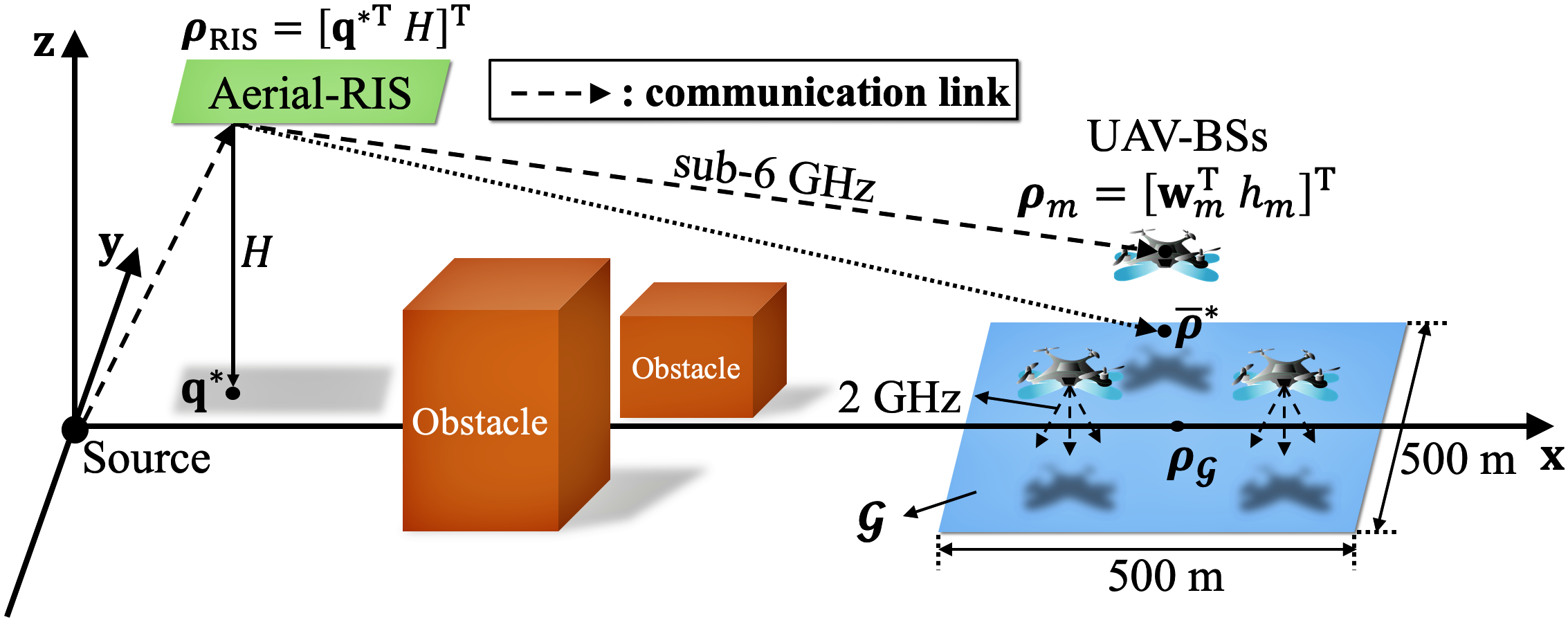}
		\caption{Simulated aerial-RIS setup comprising of a $N$-element RIS and $M_0$ UAV-BSs.}
		\label{fig_sim}
	\end{center}
\end{figure}
\section{Simulation Results}
\subsection{Simulation Setup}
For numerical analysis, we considered $10^3$ independent random user distributions and corresponding UAV-BSs~\cite{Noh} in $\mathcal{G}$ and assumed 2~GHz and sub-6~GHz channels for the fronthaul and the backhaul link, respectively~\cite{NR}. We assumed that the directional antennas installed on source are equipped with the pattern in~\cite{NR}. That is, the vertical and horizontal attenuations of the directional antenna are given by
\begin{equation}
\begin{split}
\label{antenna}
\begin{cases}
A_{\mathrm{v}} \left(\theta\right) = \min \left(12 \left( \frac{\theta-90^{\circ}}{\theta_{\mathrm{H}}} \right)^2 ,~\mathrm{SLA}_{\mathrm{v}} \right),\\
A_{\mathrm{h}} \left(\phi\right) = \min \left(12 \left( \frac{\phi}{\phi_{\mathrm{H}}} \right)^2 ,~A_{\max} \right),
\end{cases}
\end{split}
\end{equation}
where $\theta\in\left[0^{\circ}, 180^{\circ}\right]$ and $\phi\in\left[-180^{\circ}, 180^{\circ}\right)$ are the ranges of vertical and horizontal angle, respectively, $\theta_{\mathrm{H}}$ and $\phi_{\mathrm{H}}$ are vertical and horizontal HPBWs of the antenna pattern, respectively, $\mathrm{SLA}_{\mathrm{v}}$ is vertical side-lobe attenuation, and $A_{\max}$ is maximum attenuation. Hence, the antenna gain $G_{\mathrm{s}} \left(\theta, \phi\right)$ is obtained as		 
\begin{equation}
\label{ap}
G_{\mathrm{s}}\left(\theta, \phi\right)=G_{\max} -\min\left(A_{\mathrm{v}} \left(\theta\right)+A_{\mathrm{h}} \left(\phi\right) , A_{\max}\right),
\end{equation}
where $G_{\max}$ is a maximum directional gain.

For fair comparison, we considered the benchmark schemes with $(\mathbf{q}, \bar{\bm{\rho}})$ given by $\left(\frac{1}{2} \bm{\rho}_{\mathcal{G}}, \bar{\bm{\rho}}^*~\textrm{(full-array)}~\textrm{or}~\left\{\bar{\bm{\rho}}_i^*\right\}_{i=1}^L ~\textrm{(sub-array)} \right)$ and $\left(\mathbf{0}, {\bm{\rho}}_{\mathcal{G}}\right)$, respectively. The former benchmark assumes the same RIS-partition structure as our algorithm. We also considered the conventional terrestrial-based backhaul link where the ground source directly transmits the backhaul signal to the UAV-BSs based on the LoS probability and channel parameters given in~\cite{TAPUAV}.

Moreover, we analyzed the optimality of our algorithm based on exhaustive search. Since it is infeasible to search every possible $\left\{\bar{N}\right\}$, $\mathbf{q}$, and $\bm{\rho}$ that reach the global optimum of (\ref{object}), we conducted exhaustive subspace search with respect to: 1) full-array structure $\left(\bar{N}=N\right)$, 2) $\mathbf{q}\in\left\{\mathbf{q}\in\mathbb{R}^2:\left|\left|\mathbf{q}\right|\right|_2\le\delta\cdot\max_{m\in\mathbb{M}} \left|\left|\mathbf{w}_m\right|\right|_2\right\}\cap\left\{\left(x,y\right):x\ge0\right\}$ (based on~(\ref{num})), 3) $\bm{\rho}\in\mathcal{C}\cap\left\{\left(x, y, z\right): z\ge0\right\}$, where $\mathcal{C}$ is the smallest 3D cube for which every edge is parallel to x, y, or z-axis and contains every $\left\{\bm{\rho}_m\right\}_{m\in\mathbb{M}}$. Since adopting a full-array structure, placing aerial-RIS near the backhaul source, and setting the phase align point close to every UAV-BS lead to better performance, we believe that our subspace search indeed leads to a reliable suboptimal solution. The simulation environment based on the parameters is illustrated in Fig.~\ref{fig_sim}, and the detailed parameters are given in Table~\ref{SimPar}.
\begin{figure}[t]
	\begin{center}
		\includegraphics[width=0.97\columnwidth,keepaspectratio]%
		{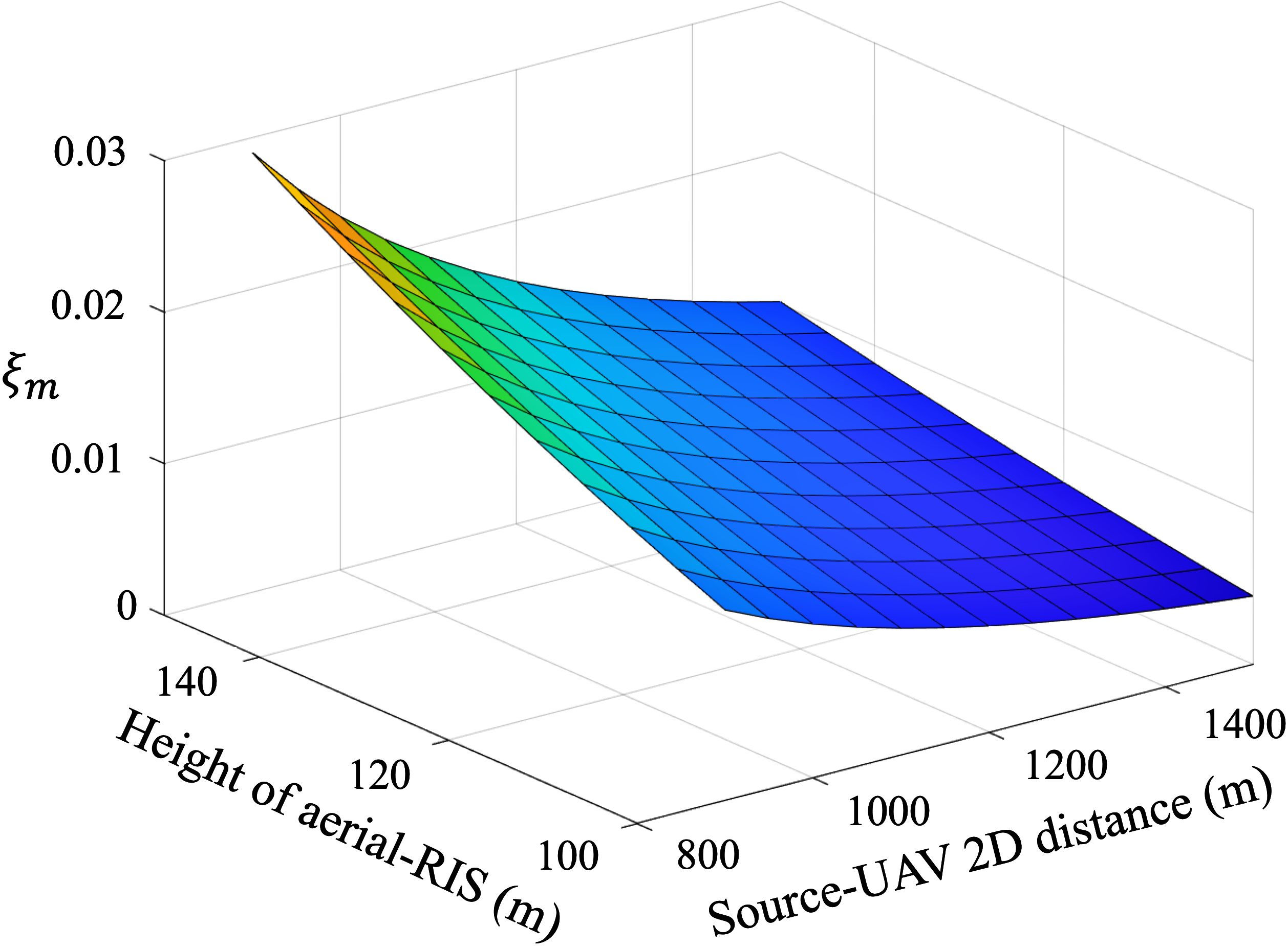}
		\caption{Behavior of $\xi_m$ with $h_m = 45~\textrm{m}$ with respect to the height of aerial-RIS and the source-UAV 2D distance.}
		\label{fig_cubic}
	\end{center}
\end{figure}
	\begin{figure}[t]
	\begin{center}
		\includegraphics[width=0.97\columnwidth,keepaspectratio]%
		{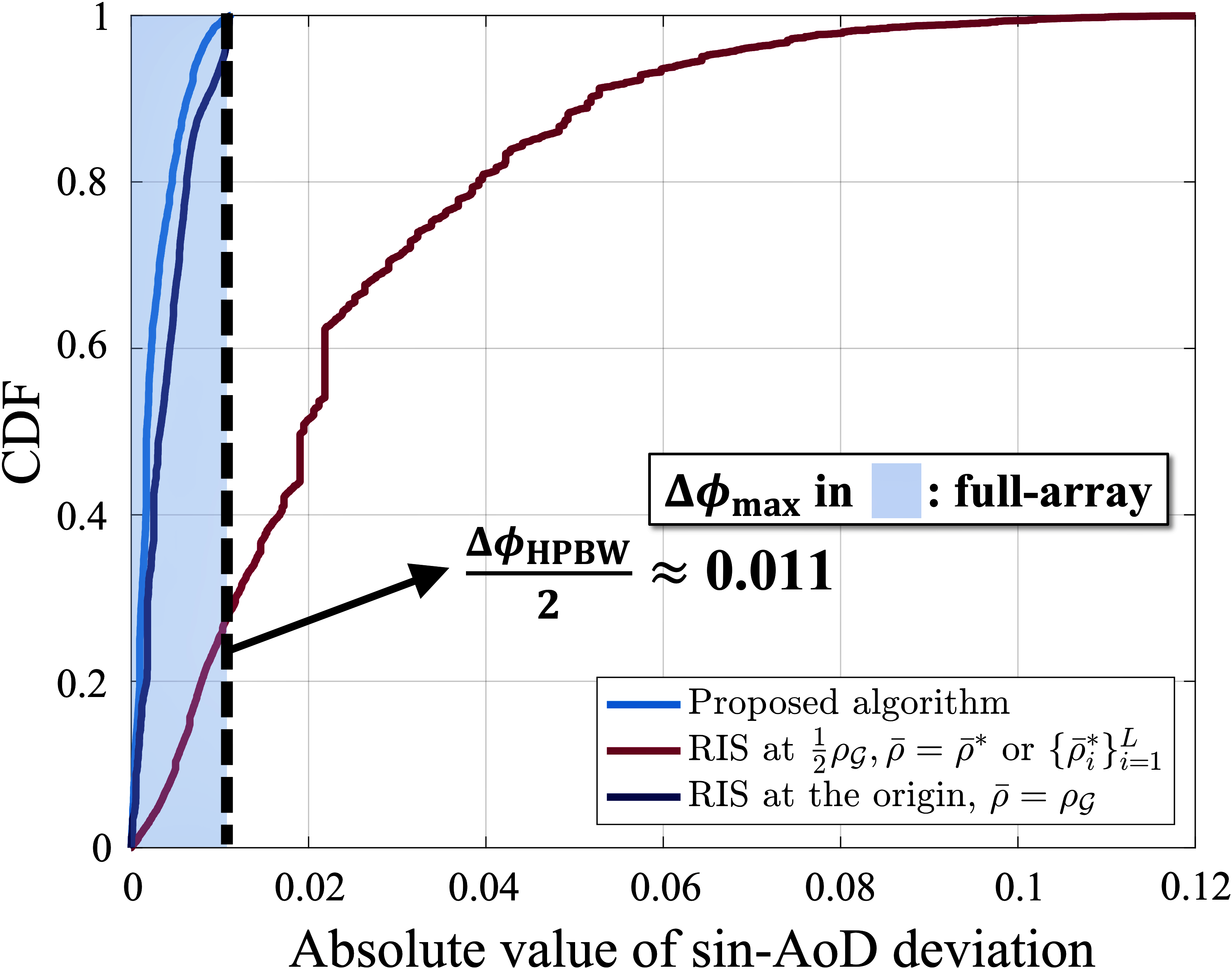}
		\caption{CDF of the absolute value of the sin-AoD deviations with the phase align point $\bm{\bar\rho}^*$ $\left(N=400\right)$.}
		\label{fig_sin}
	\end{center}
\end{figure}
	 \begin{center}
	\begin{table}[t]
	\centering 
		\caption{Empirical Probability of $\left|\Delta\phi_m (\bm{\bar\rho}^*)\right|> \frac{\Delta\phi_{\mathrm{HPBW}}}{2}$ in Fig.~\ref{fig_sin}}
		\begin{tabular}{|>{\centering} m{4.9cm} |>{\centering} m{3cm}|}
			\hline
			\textbf{Scenario} & \textbf{Probability}
			\tabularnewline
			\hline
			\centering	\textbf{Proposed aerial-RIS setup} & $\mathbf{7.50\times10^{-3}}$  \tabularnewline \hline
			\centering	RIS at $\frac{1}{2}\rho_{\mathcal{G}}, \bm{\bar\rho}=\bm{\bar\rho}^*~\textrm{or}~\left\{\bm{\bar\rho}_i^*\right\}_{i=1}^L$  & $7.22\times10^{-1}$   \tabularnewline \hline
			\centering	RIS at the origin, $\bm{\bar\rho}=\bm{\bar\rho}_{\mathcal{G}}$ & $2.24\times10^{-2} $  \tabularnewline \hline
		\end{tabular}
		\label{ProbT}
	\end{table}
\end{center}
\begin{figure*}[t]
	\begin{center}
		\includegraphics[width=1.99\columnwidth,keepaspectratio]%
		{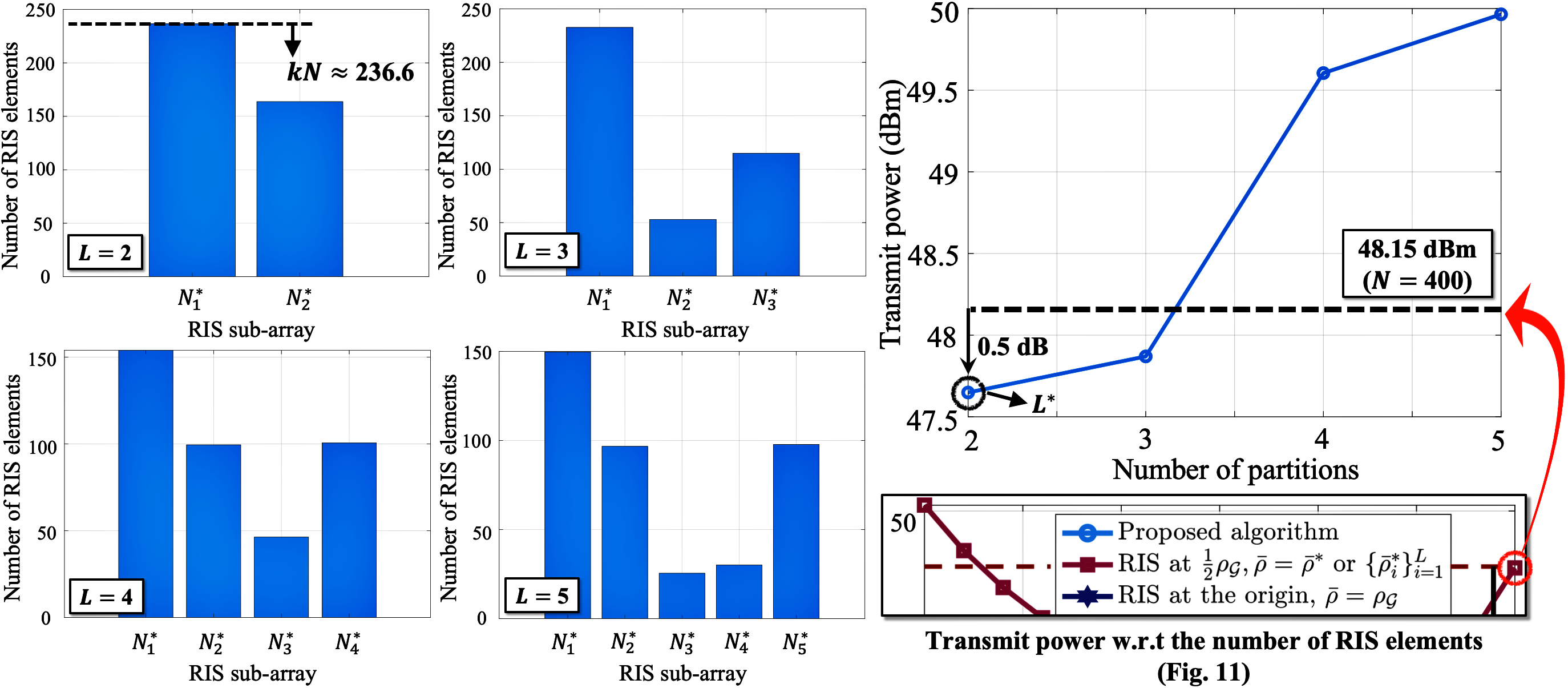}
		\caption{Array-partition result with $N=400$ and the corresponding transmit power for each number of partitions.}
		\label{figArray}
	\end{center}
\end{figure*}
\begin{figure}[t]
	\begin{center}
		\includegraphics[width=0.94\columnwidth,keepaspectratio]%
		{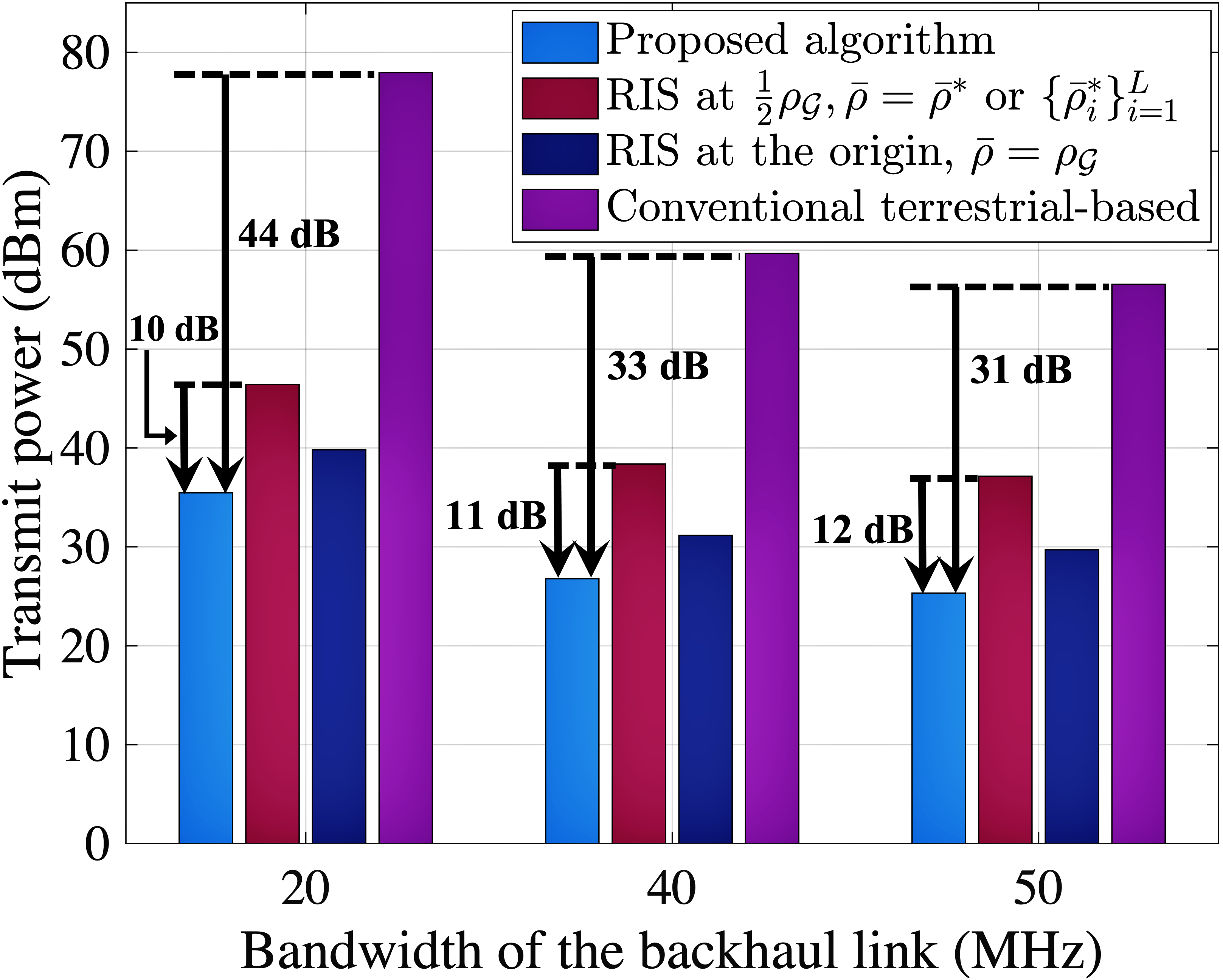}
		\caption{Transmit power with respect to the bandwidth of the backhaul link.}
		\label{figBW}
	\end{center}
\end{figure}
\subsection{Reliability of Aerial-RIS Placement and Phase Alignment}
Fig.~\ref{fig_cubic} shows the behavior of $\xi_m$ according to the height of aerial-RIS $\left(H\right)$ and the 2D distance of the UAV-BS from the source $\left(\left|\left|\mathbf{w}_m\right|\right|_2\right)$. It is shown that $\xi_m$ has an order of $10^{-2}$ for sufficiently large $\left|\left|\mathbf{w}_m\right|\right|_2$, which implies that $\mathbf{q}_m^*$ is extremely close to the origin compared to $\mathbf{w}_m$. Besides, $\xi_m$ increases according to $H$. It is because the increase of $H$ leads to larger $a$ and smaller $\sqrt{-\frac{a}{3}}~\left(a<0\right)$. Furthermore, because the posterior term of the $\sqrt{-\frac{a}{3}}$ in~(\ref{xixixixi}), which is also approximated by~$\left(-1+\frac{\epsilon}{3\sqrt{3}} \right)~\left(|\epsilon|\ll1\right)$ in~(\ref{taylor}), is negative, we can interpret that the decrease in the term $\sqrt{-\frac{a}{3}}$ significantly affects  the increase of $\xi_m$. It also coincides with \textit{Remark}~\ref{r1}, where $\xi_m$ increases and converges to $\frac{1}{2}$ when $H\rightarrow\frac{1}{2} \left|\left|\mathbf{w}_m\right|\right|_2$. Moreover, as $\xi_m \ll 1$, $\Delta\phi_m\left(\cdot\right)$ stays nearly the same as the $\mathbf{q}_m^*=\mathbf{0}$ case, where the numerically evaluated difference is given by less than $10^{-4}$, and thus negligible for the same environment. Nevertheless, the proposed method leads to energy-efficiency enhancement compared to the case of simply letting $\mathbf{q}_m^*=\mathbf{0}$. 

Fig.~\ref{fig_sin} illustrates the cumulative distribution function (CDF) of the absolute value of the sin-AoD deviations between the UAV-BS and the phase align point $\bm{\bar\rho}^*$ with $N=400$. We can apply the full-array RIS structure, which leads to the maximum energy-efficiency, when every absolute value of sin-AoD deviation is less than $\frac{\Delta\phi_{\mathrm{HPBW}}}{2}\approx \frac{1}{2} \frac{0.8858}{400*0.1}\approx0.011$. As shown from the CDF plot, the proposed algorithm outperforms every benchmark distributed in ``$\le\frac{\Delta\phi_{\mathrm{HPBW}}}{2}$" region colored in blue for almost 100\% of probability, where for the worst case of the benchmarks, only about 30\% of them are contained in the region. For more accurate analysis, we have also computed the empirical probability that the absolute value exceeds $\frac{\Delta\phi_{\mathrm{HPBW}}}{2}$ based on the CDF and organized in Table~\ref{ProbT}. As we have conducted the numerical simulation for $10^3$ independent user distributions with more than 8,000~UAV-BSs in total, from the result, we can conclude that the excess of the absolute value occurs for about 60~UAV-BSs by the proposed algorithm, which implies that the proposed algorithm guarantees the full-array scenario almost completely. As for the benchmarks, the value is larger than the proposed algorithm in every case. For the worst scenario, the probability of excess exceeds 0.5. Therefore, it is forced to apply the sub-array architecture and leads to the loss of energy-efficiency.

\subsection{Transmit Power Comparisons under Several Effects}
Fig.~\ref{figArray} illustrates the array-partition result $\left(\left\{N_i^*\right\}_{i=1}^L\right)$ with $N=400$ and the transmit power of the source corresponding to each number of partitions. We considered the number of partition $L$ by $L=2, \cdots, L_{\max}\left(=5\right)$ and compute the optimal array-partition for each $L$ by~(\ref{optimalssss}). We can find that the upper-bound of $\left\{N_i^*\right\}_{i=1}^L \left(kN\right)$ is meaningful at $L=2$ case, because the increase of $L$ leads to the rise of $k=L \frac{\Delta\phi_{\mathrm{HPBW}}/2}{\Delta\phi_{\max}}$ in~(\ref{constN}), reducing the significance of the upper-bound $kN$. Additionally, It is observed that the transmit power increases according to the number of array-partitions. The result agrees with our derived result in that, the array-partition leads to the decrease of the peak gain of $g$ and the increase of the transmit power. Moreover, by searching $L^*$, the proposed algorithm can achieve even less transmit power for approximately 0.5~dB compared to the benchmark with the same $N=400$~case, which is given by 48.15~dBm.

Fig.~\ref{figBW} shows the transmit power according to the bandwidth of the backhaul link $\left(B_{\mathrm{b}}\right)$. It is clear that when the bandwidth increases, by~(\ref{rate}), we can achieve the same backhaul rate $\left\{R_m\right\}_{m\in\mathbb{M}}$ with less transmit power, thereby increasing energy-efficiency. In fact, due to the efficient setup of the aerial-RIS by the proposed algorithm, the transmit power decreases further than the benchmarks, where the gain is up to 12~dB for every scenario with RIS. Moreover, the gain is up to 44~dB for terrestrial backhaul scenario without RIS, which suffers from extreme path loss due to high NLoS probability~\cite{a2gglobecom, TAPUAV} and leads to infeasible transmit power.

Fig.~\ref{figNum} compares the transmit power according to the number of RIS elements $\left(N\right)$. Here, we considered the upper-bound of $N$ due to the area limitation of the aerial platform and the high possibility of stability loss of the aerial platform that carries a vast number of RIS elements~\cite{traj}. 
As shown in the figure, the proposed method significantly reduces the transmit power compared to the benchmarks, where the gain is up to 10$\sim$13~dB for RIS-equipped cases and up to 32~dB for the terrestrial backhaul scenario, which suffers from high NLoS component. When $N$ goes significantly higher, in fact, the HPBW is reduced by the order of $N^{-1}$, which leads to the decrease of the probability of using the full-array structure and consequently causes the increase of a transmit power. Nevertheless, even for large $N$, our algorithm guarantees reliable transmit power, whereas for the benchmarks the power continues to increase. We can therefore conclude that for the feasible range of $N$ that we have considered, the performance of our algorithm monotonically increases with the number of RIS elements.

Fig.~\ref{figCenter} shows the transmit power according to the distance of $\bm{\rho}_{\mathcal{G}}$ from the source $\left(d_{\mathcal{G}}\right)$. It is clear that by applying the proposed algorithm, we can achieve an additional reduction of transmit power compared to the benchmarks. Even for low $d_{\mathcal{G}}$, our algorithm guarantees reliable transmit power, where for RIS-equipped benchmarks the power increases due to the larger possibility of using sub-array structure 
or the ``HPBW-outlier" by simply letting $\bm{\bar\rho} = \bm{\rho}_{\mathcal{G}}$.
\begin{figure}[t]
	\begin{center}
		\includegraphics[width=0.97\columnwidth,keepaspectratio]%
		{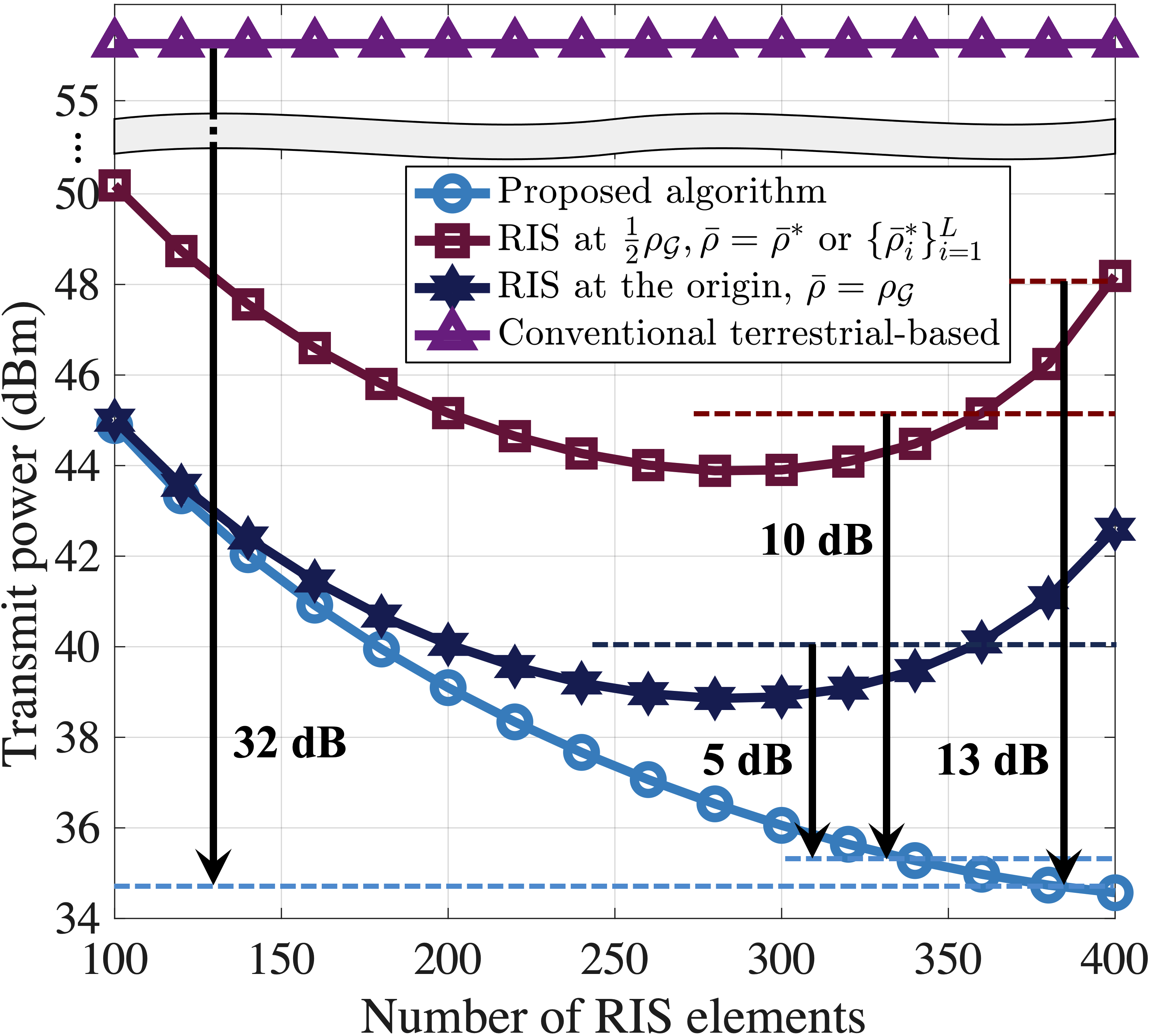}
		\caption{Transmit power with respect to the number of RIS elements.}
		\label{figNum}
	\end{center}
\end{figure}

Fig.~\ref{figheight} illustrates the transmit power according to feasible $H$ with the guarantee of high LoS probability~\cite{a2gglobecom, Noh}. Clearly, by applying our algorithm, we can reduce the transmit power by 6$\sim$7~dB compared to the RIS-equipped benchmark algorithms, and by about 23~dB compared to the terrestrial backhaul solution. Although the power increases for high $H$ owing to the increase of path loss, our algorithm can maintain a reliable level of transmit power comparable with the low-$H$ scenario. The power growth becomes larger for the benchmarks, resulting in lower energy-efficiencies. Hence, we can conclude that we should choose sufficiently low $H$ with the consideration of the height of skyscrapers in the urban area and the maintenance of the high probability of LoS links.

Fig.~\ref{figopt} shows a comparison between the solution obtained by the proposed algorithm and the suboptimal solution obtained from the exhaustive subspace search according to $d_{\mathcal{G}}$ and two different values of $H$. Evidently, owing to the alternative optimization and considered approximations, the proposed solution has a slight error with the suboptimal solution. However, the maximal optimality gap of the simulation is about 1$\%$, which is negligible. Therefore, we can deduce that there is no significant impact on the performance of energy-efficiency.

\begin{figure}[t]
	\begin{center}
		\includegraphics[width=0.91\columnwidth,keepaspectratio]%
		{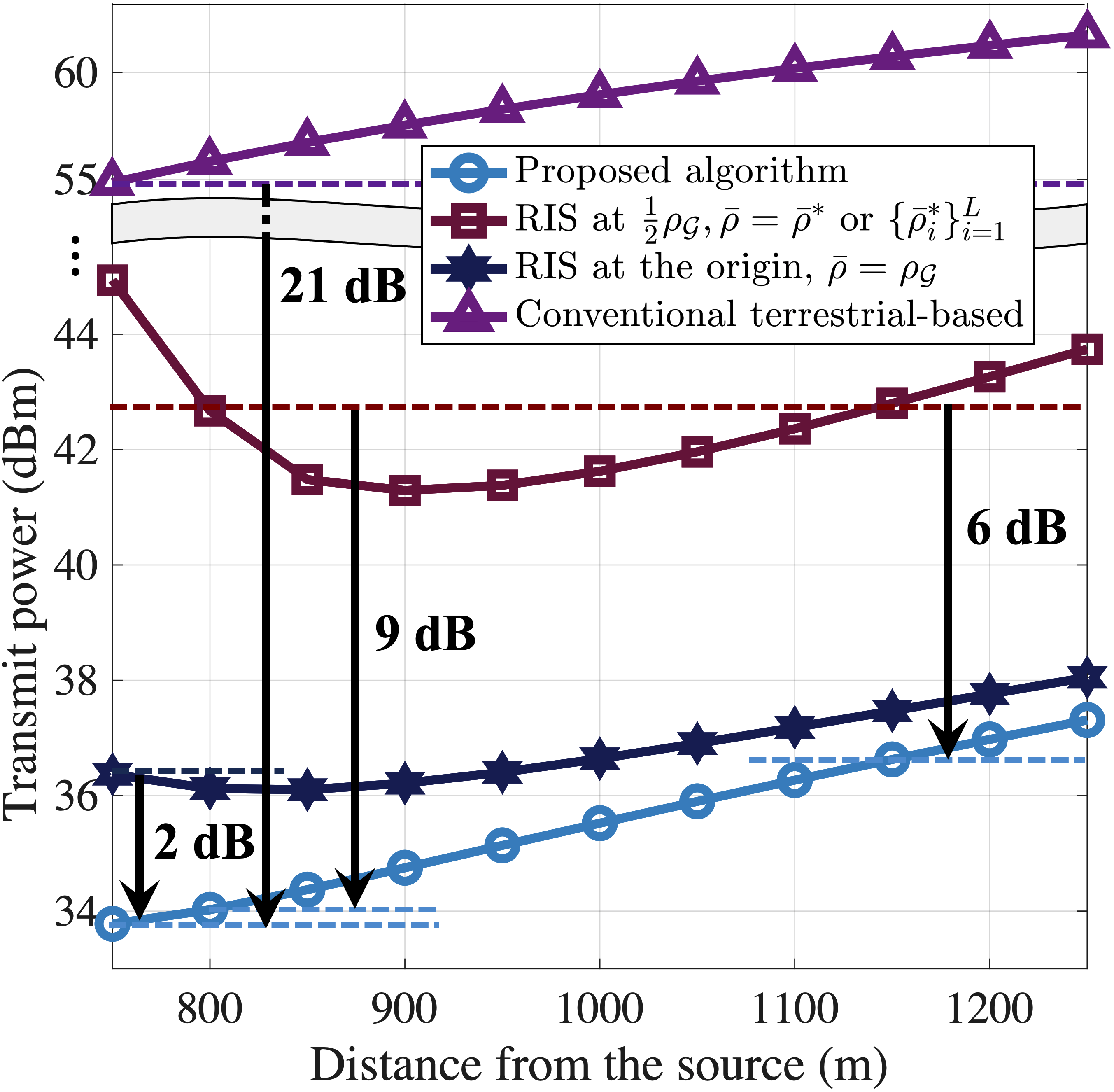}
		\caption{Transmit power with respect to the distance of $\bm{\rho}_{\mathcal{G}}$ from the source.}
		\label{figCenter}
	\end{center}
\end{figure}
\section{Conclusion}
In this paper, we have proposed an aerial-RIS backhaul architecture that can support UAV-BSs with a high energy-efficiency. We have derived the lower-bound on source transmit power for backhauling UAV-BSs for a given data rate and minimized the lower-bound. Through balancing the beamforming gain of aerial-RIS and the product of squared distances between multiple UAV-BSs by minimizing the sum of $\ell_2$-deviations and considering the proposed ``reverse" waterfilling array-partition, we determined the placement of aerial-RIS and the suboptimal phase align point/points that significantly enhance energy-efficiency. Additionally, the total computational complexity does not exceed the quadratic order, which also guarantees the efficiency of computation. Simulation results confirm that, compared to the benchmark schemes, the proposed method can achieve additional energy-efficiency by ensuring a source-nearby aerial-RIS placement and consequently a high probability of selecting the full-array RIS structure. We anticipate our analysis providing a deep insight for successfully integrating RIS with high-altitude platform and the performance analysis concerning several parameters, including the 3D placement and the phase align point/points of aerial-RIS.

For future work, our research can be broadened by extending the RIS structure, i.e., the UPA structure~\cite{UPARIS, XLRIS}, and considering the physical and hardware characteristics of the RIS dependent on the size and number of the RIS elements~\cite{meta, holo, LingRIS, TMH}. Furthermore, one can design the aerial-RIS trajectory to adapt to moving UAV-BSs under the mobility scenario~\cite{SRS, drl, HJFSO, FSOJSAC}. It is also essential to study other factors that affect the energy-efficiency, such as the battery consumption of the aerial platforms~\cite{traj, battery2}, assistance of the multiple RIS deployment including the terrestrial RIS~\cite{How, IRSmeetsUAV}, and the power dissipation in the RIS hardware~\cite{RISEE, DF}. Moreover, the work on the methodology of RIS installation on the aerial platform can be performed~\cite{UAVRISmag}, which is still in its infancy. Comparative research of our proposed aerial-RIS setup algorithm and other relaying methods, including amplify-and-forward (AF)~\cite{RISEE} and decode-and-forward (DF) architectures~\cite{vsrelay, DF, RISEE} can be taken into account. Furthermore, to confront the higher-data-rate fronthaul/backhaul links for the future wireless network~\cite{map, MS, kwon2016rf_LE, IAB, ygtwc, XR, Tact, HBFD, Squint}, the application of signals with a higher spectrum can be contemplated for our proposed aerial-RIS backhaul system~\cite{drl, RIStera, RISFSO}, followed by the software simulation and prototype measurement of its performance and reliability~\cite{MK,MS2,MKA}.	
		\begin{figure}[t]
	\begin{center}
		\includegraphics[width=0.93\columnwidth,keepaspectratio]%
		{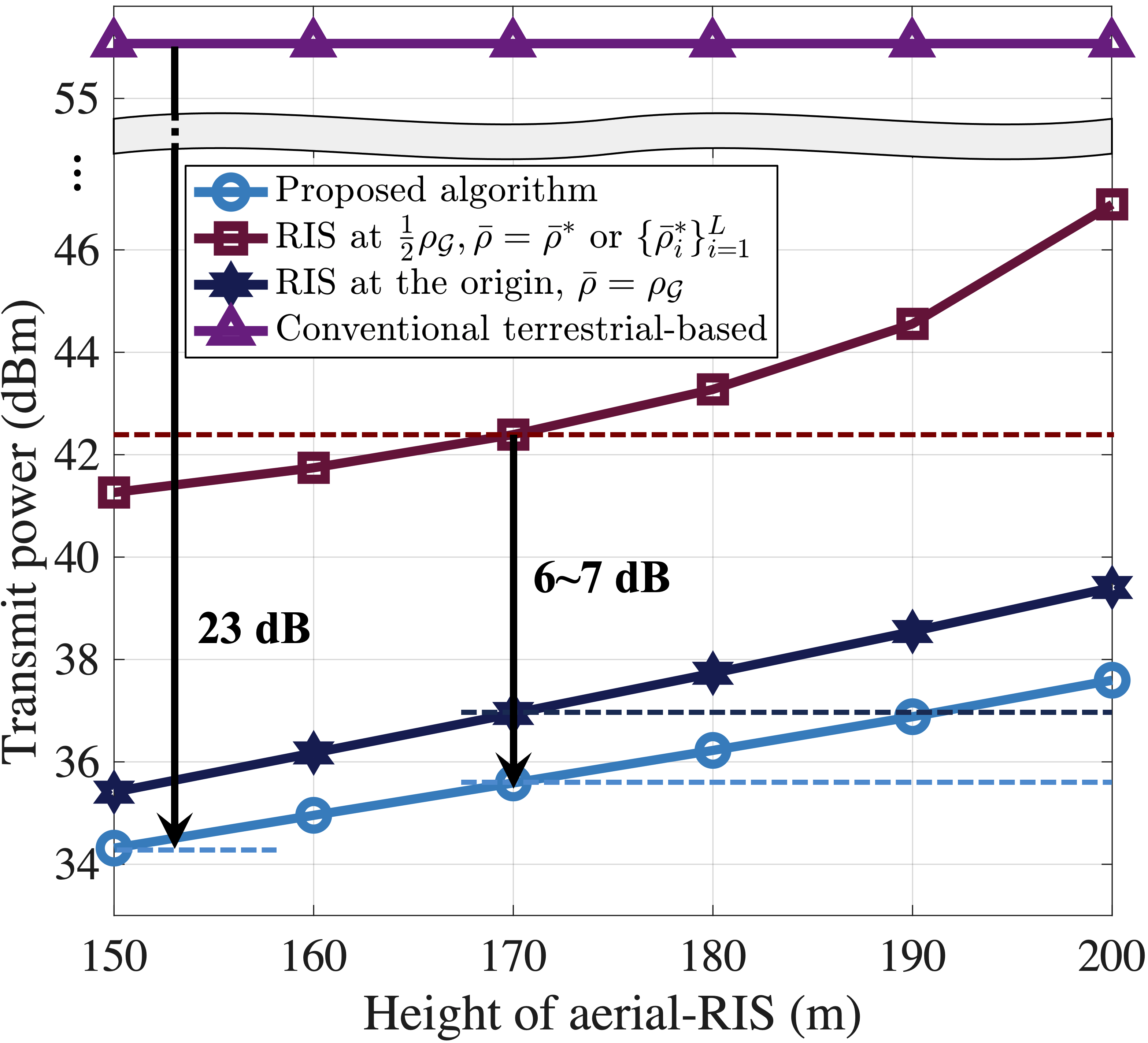}
		\caption{Transmit power with respect to the height of aerial-RIS.}
		\label{figheight}
	\end{center}
\end{figure}
\begin{figure}[t]
	\begin{center}
		\includegraphics[width=0.93\columnwidth,keepaspectratio]%
		{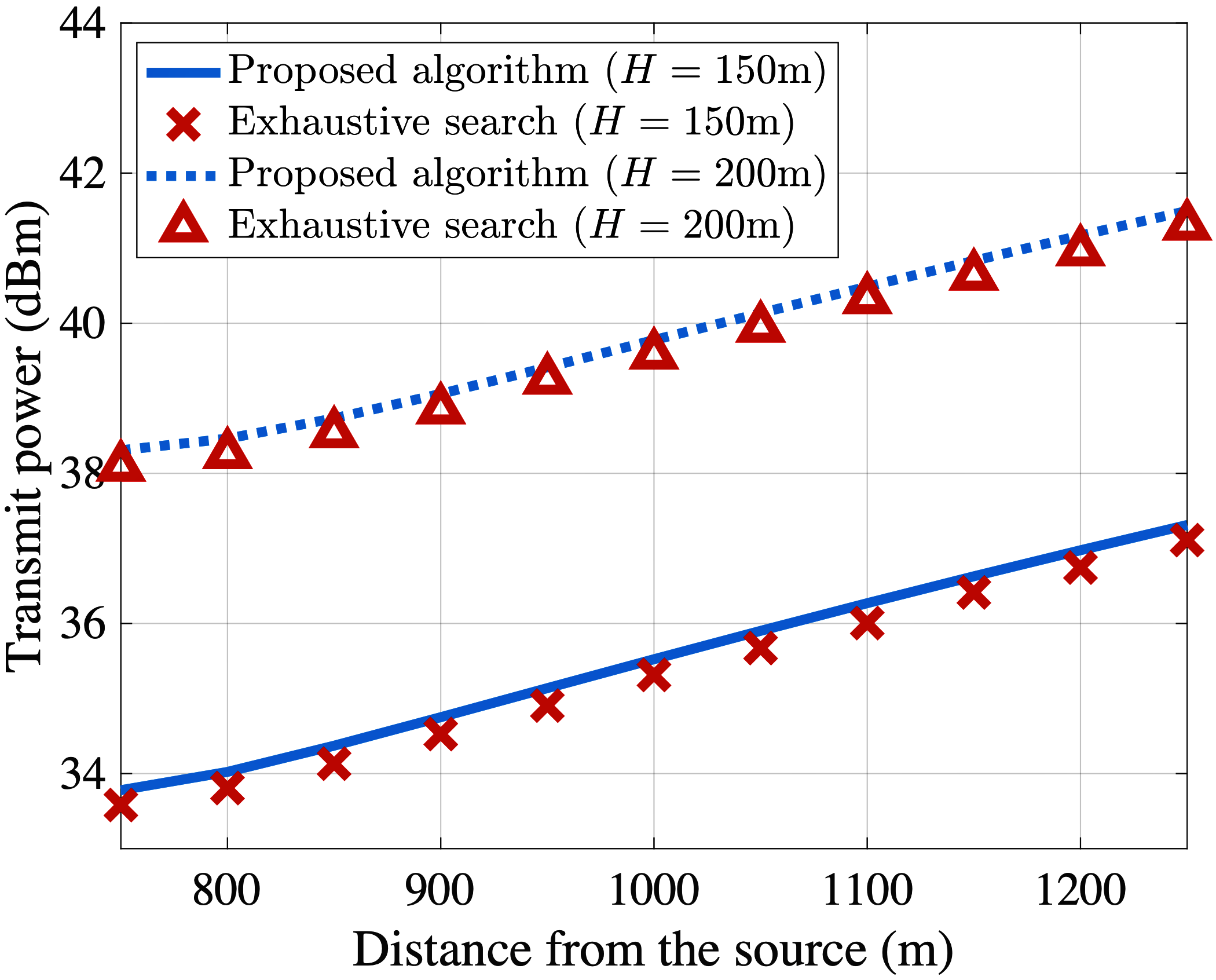}
		\caption{The optimality validation of the proposed solution with respect to $d_{\mathcal{G}}$ for $H=150, 200~\textrm{m}$.}
		\label{figopt}
	\end{center}
\end{figure}
	\section*{Acknowledgement}
	The authors would like to thank Prof. Robert W. Heath Jr., Dr. Min Soo Sim and Mr. Yonghwi Kim for their helpful discussions.
			\begin{figure}[t]
	\begin{center}
		\includegraphics[width=0.98\columnwidth,keepaspectratio]%
		{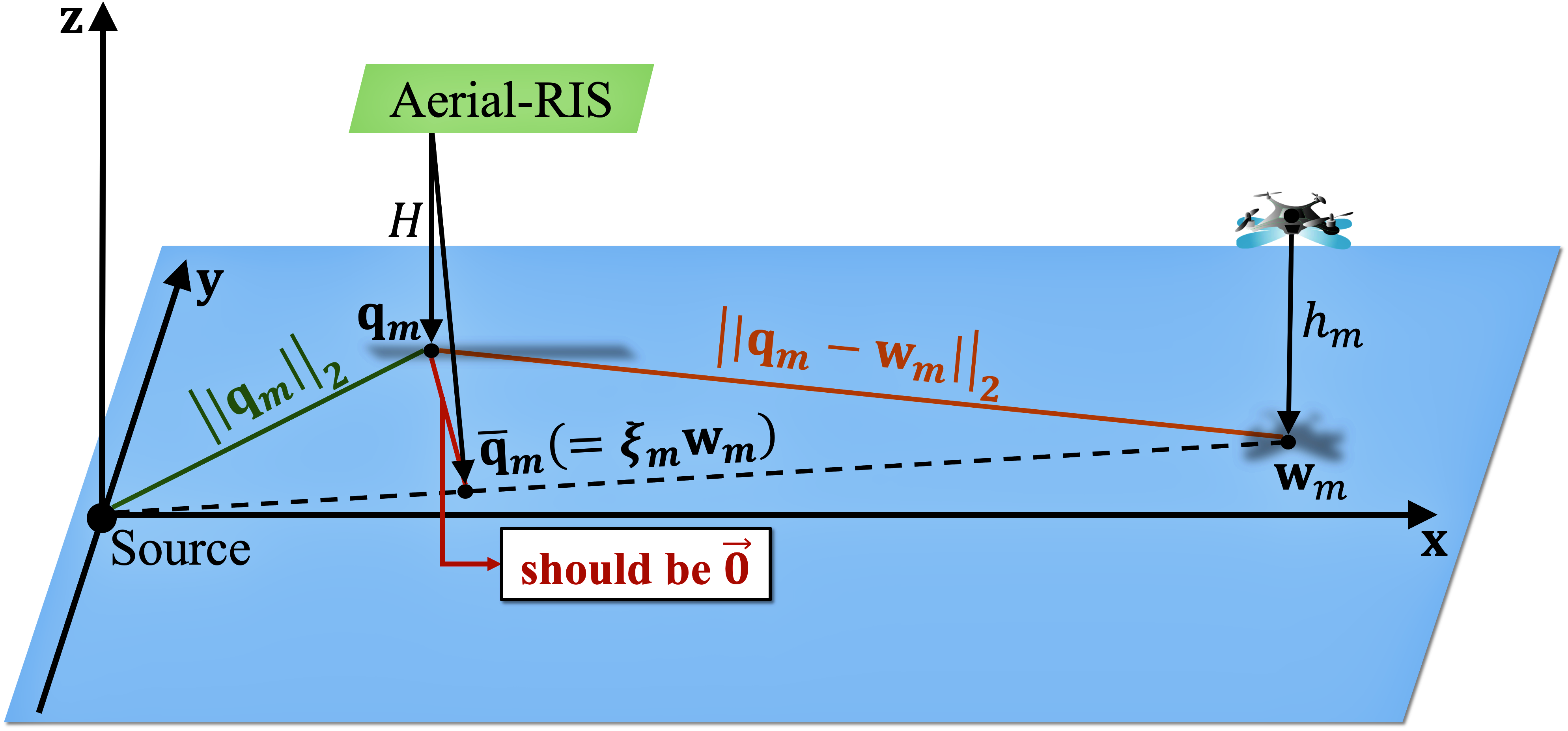}
		\caption{The orthogonal projection of $\mathbf{q}_m$ and illustration of~(\ref{projcond}) in Theorem~\ref{cubiceq}.}
		\label{figth2}
	\end{center}
\end{figure}
		\section*{Appendix A}
\section*{Proof of Theorem~\ref{cubiceq}}
	Let the orthogonal projection of $\mathbf{q}_m$ onto the line segment that links the source in the origin and $\mathbf{w}_m$ be $\mathbf{\bar{q}}_m$. Then, by the definition of the orthogonal projection, the following holds:
	\begin{equation}
	\begin{split}
	\label{projq}
	\begin{cases}
	\left|\left|\mathbf{q}_m - \mathbf{w}_m\right|\right|_2^2 = \left|\left|\mathbf{q}_m - \mathbf{\bar{q}}_m\right|\right|_2^2 + \left|\left|\mathbf{\bar{q}}_m - \mathbf{w}_m\right|\right|_2^2,\\
	\left|\left|\mathbf{q}_m \right|\right|_2^2 = \left|\left|\mathbf{q}_m - \mathbf{\bar{q}}_m\right|\right|_2^2 +\left|\left|\mathbf{\bar{q}}_m \right|\right|_2^2.
	\end{cases}
	\end{split}
	\end{equation}	
	Here, to minimize the objective function of~(\ref{num}), we need to minimize the left-hand side of both equations in~(\ref{projq}). Hence, we have to choose $\mathbf{q}_m$ that satisfies the following:
	\begin{equation}
	\label{projcond}
	\left|\left|\mathbf{q}_m - \mathbf{\bar{q}}_m\right|\right|_2^2 =0~\leftrightarrow~\mathbf{q}_m=\xi_m \mathbf{w}_m~\left(\xi_m>0\right),
	\end{equation}
which is clarified in Fig.~\ref{figth2}. Hence, the solution $\mathbf{q}_m^*$ of the problem is given by  $\mathbf{q}_m^*=\xi_m \mathbf{w}_m$, and by substituting we can denote the objective function by $f\left(\xi_m\right)$, that is,
	\begin{equation}
		\label{fxi}
		f\left(\xi_m\right)\triangleq \left|\left|\mathbf{w}_m\right|\right|_2^4 \left(\xi_m^2 + \zeta_1^2 \right) \left(\left(1-\xi_m\right)^2 + \zeta_2^2 \right),
	\end{equation}
where $\zeta_1=\frac{H}{\left|\left|\mathbf{w}_m\right|\right|_2}$ and $\zeta_2 = \frac{\left|H-h_m\right|}{\left|\left|\mathbf{w}_m\right|\right|_2}$. To find the minimum of $f\left(\xi_m\right)$, we should find the root of $f'\left(\xi_m\right)$. The discriminant $\Delta$ of the cubic equation $f'\left(\xi_m\right)=0$ is given by~\cite{Lovett}
\begin{equation}
	\label{disc}
	\Delta = \left(\frac{a}{3}\right)^3 + \left(\frac{b}{2}\right)^2,
\end{equation}
where $a=\frac{1}{2} \left(\zeta_1^2 + \zeta_2^2 \right) - \frac{1}{4}$ and $b=\frac{1}{4} \left(\zeta_2^2 - \zeta_1^2 \right)$. Because we assume that $d_{\mathcal{G}}$ is sufficiently large, it is reasonable to assume that $\left|\left|\mathbf{w}_m\right|\right|_2$ follows the same scale, which leads to $a<0\left(\approx-\frac{1}{4}\right)$ and $|b|\ll1$. In this case, $\Delta<0$ holds and the three real solutions $\left\{\xi_{m,k}\right\}_{k=0}^2$ of $f'\left(\xi_m\right)=0$ are given by~\cite{Lovett}
	\begin{equation}
		\begin{split}
		\label{solcubic}
		\xi_{m,k}=\frac{1}{2}+2\sqrt{-\frac{a}{3}}\cos \left( \frac{1}{3} \cos^{-1} \left( \frac{3b}{2a} \sqrt {-\frac{3}{a}} \right)-\frac{2}{3}\pi k \right)&\\
		\left(k=0,1,2\right)&.
	\end{split}
	\end{equation}
	As $a\approx-\frac{1}{4}$ and $|b|\ll 1$, by letting $\frac{3b}{2a} \sqrt{-\frac{3}{a}} = \epsilon ~\left(|\epsilon|\ll 1\right)$ and successively applying the first-order Taylor approximation to~(\ref{solcubic}), it becomes
	\begin{equation}
	\begin{split}
	\label{taylor}
	\begin{cases}
	\xi_{m,0} \approx \frac{1}{2}+\sqrt{-a} \left(1+ \frac{\epsilon}{3\sqrt{3}} \right),\\
	\xi_{m,1} \approx \frac{1}{2}+\sqrt{-a} \left(- \frac{\epsilon}{3\sqrt{3}} \right),\\
	\xi_{m,2} \approx \frac{1}{2}+\sqrt{-a} \left(-1+ \frac{\epsilon}{3\sqrt{3}} \right).
	\end{cases}
	\end{split}
	\end{equation}
	Moreover, by substituting $a=\frac{1}{2} \left(\zeta_1^2 + \zeta_2^2\right) - \frac{1}{4}$ into~(\ref{taylor}) we~obtain
		\begin{equation}
	\begin{split}
	\label{taylor2}
	\begin{cases}
	\xi_{m,0} \approx \frac{1}{2}-\sqrt{\frac{1}{4}-\frac{1}{2} \left(\zeta_1^2 + \zeta_2^2 \right)} \left(-1- \frac{\epsilon}{3\sqrt{3}} \right),\\
	\xi_{m,1} \approx \frac{1}{2}-\sqrt{\frac{1}{4}-\frac{1}{2} \left(\zeta_1^2 + \zeta_2^2 \right)} \left( \frac{\epsilon}{3\sqrt{3}} \right),\\
	\xi_{m,2} \approx \frac{1}{2}-\sqrt{\frac{1}{4}-\frac{1}{2} \left(\zeta_1^2 + \zeta_2^2 \right)} \left(1- \frac{\epsilon}{3\sqrt{3}} \right).\\
	\end{cases}
	\end{split}
	\end{equation}	
	From~(\ref{taylor2}), we can obtain the following outcomes:
	 \begin{enumerate}
 \item It is clear that $\xi_{m,0}>\xi_{m,1}>\xi_{m,2}$ for small $|\epsilon|$. Owing to the characteristics of the quartic equation~\cite{Lovett}, $f$ achieves its local minimum at $\xi_{m,0}$ and $\xi_{m,2}$, and either of two achieves the global minimum.
 \item Since $|\epsilon|\ll1$, the following holds for $\xi_{m,2}$:
 \begin{equation}
 \label{taylor3}
 \xi_{m,2}\approx \frac{1}{2}-\sqrt{\frac{1}{4} -\frac{1}{2} \left(\zeta_1^2 + \zeta_2^2 \right)} \left(1\right) =\frac{\frac{1}{2} \left(\zeta_1^2 + \zeta_2^2 \right)}{\sqrt{\frac{1}{4} + \frac{1}{2} \left(\zeta_1^2 + \zeta_2^2 \right)}},
 \end{equation}
which implies that $\xi_{m,2}$ is greater than 0 for small $|\epsilon|$. Moreover, since $\zeta_1, \zeta_2$ is sufficiently small ($\because\left|\left|\mathbf{w}_m\right|\right|_2$ is sufficiently large), by~(\ref{taylor3}) it can be deduced that $\xi_{m,2}$ is close to the origin. 
 \end{enumerate}
From above results, we should choose $\xi_m$ as
	\begin{equation}
	\label{cubicres}
	\xi_m\triangleq\xi_{m,2}=\frac{1}{2}+2\sqrt{-\frac{a}{3}}\cos \left( \frac{1}{3} \cos^{-1} \left( \frac{3b}{2a} \sqrt {-\frac{3}{a}} \right)-\frac{4}{3}\pi  \right),
	\end{equation}
	and determine $\mathbf{q}_m^*=\xi_m\mathbf{w}_m$, which is the result of the~theorem.~$\blacksquare$
			\section*{Appendix B}
\section*{Proof of Theorem~\ref{KKT}}
The Lagrangian $\mathcal{L}$ of Problem~(\ref{water}) with dual variable $\left[\left\{\lambda_{1i}\right\}_{i=1}^L, \left\{\lambda_{2i}\right\}_{i=1}^L, \mu\right]$ is given by~\cite{boyd}
\begin{equation}
\begin{split}
\label{LagN}
&\mathcal{L}\left(\left\{\lambda_{1i}\right\}_{i=1}^L, \left\{\lambda_{2i}\right\}_{i=1}^L, \mu, \left\{N_i \right\}_{i=1}^L\right)\\&=\sum_{i=1}^L\sum_{m\in\mathbb{M}_i} \frac{A_m}{N_{i}^2} + \sum_{i=1}^L \lambda_{1i} \left(N_i - kN\right) \\&~~~+ \sum_{i=1}^L \lambda_{2i} \left(-N_i \right) + \mu \left(\sum_{i=1}^L N_i -N \right).
\end{split}
\end{equation}
Hence, by differentiating $\mathcal{L}$ into $N_i$ and applying the Karush-Kuhn-Tucker (KKT) conditions~\cite{boyd} we obtain
\begin{equation}
\begin{split}
\label{diffL}
\begin{cases}
\frac{\partial \mathcal{L}}{\partial N_i} = \sum_{m \in\mathbb{M}_i} \frac{-2A_m}{N_i^3} +  \left(\lambda_{1i} -\lambda_{2i}\right)+\mu =0~\left(\star\right),\\
\lambda_{1i}\ge0, \lambda_{2i} \ge0, \lambda_{1i} \left(N_i - kN\right)=0 ,  \lambda_{2i} N_i =0.
\end{cases}
\end{split}
\end{equation}
However, since every RIS partition, in fact, should have a positive number of elements $\left(N_i\neq0\right)$, $\lambda_{2i} =0$ holds. For the case of $\lambda_{1i}=0$, $\left(\star\right)$ in (\ref{diffL}) is equivalent to
\begin{equation}
\begin{split}
\label{muu}
 &\frac{\partial \mathcal{L}}{\partial N_i}=\sum_{m \in\mathbb{M}_i} \frac{-2A_m}{N_i^3} +\mu =0\\
 &~\rightarrow N_i = \sqrt[3]{\sum_{m\in\mathbb{M}_i} {\frac{2A_m}{\mu}}}\le kN.
\end{split}
\end{equation}
However, for the case of $N_i-kN=0$, by substituting~$N_i=kN$ into $\left(\star\right)$ we obtain
\begin{equation}
\begin{split}
\label{bigsmall}
&\mu + \sum_{m\in\mathbb{M}_i} \frac{-2A_m}{\left(kN\right)^3 } = -\lambda_{1i} \le0\\
&~\rightarrow N_i=kN\le\sqrt[3]{\sum_{m\in\mathbb{M}_i} {\frac{2A_m}{\mu}}}.
\end{split}
\end{equation}
Hence, from~(\ref{muu}) and~(\ref{bigsmall}), the optimal $\left\{N_i^*\right\}_{i=1}^L$ is given by
\begin{equation}
\label{optimalN}
 N_i^*=\min\left(kN, \sqrt[3]{\frac{\sum_{m\in\mathbb{M}_i} 2A_m}{\mu}}\right) \left(i=1,\cdots,L\right),
 \end{equation}
 where $\mu$ is determined so that $\sum_{\ell=1}^L N_\ell^* =N$ holds and the theorem follows.~$\blacksquare$	
		
		\bibliographystyle{IEEEtran}
	     \bibliography{RIS_ref}	

\begin{IEEEbiography}
    [{\includegraphics[width=1in,height=1.25in,clip,keepaspectratio]{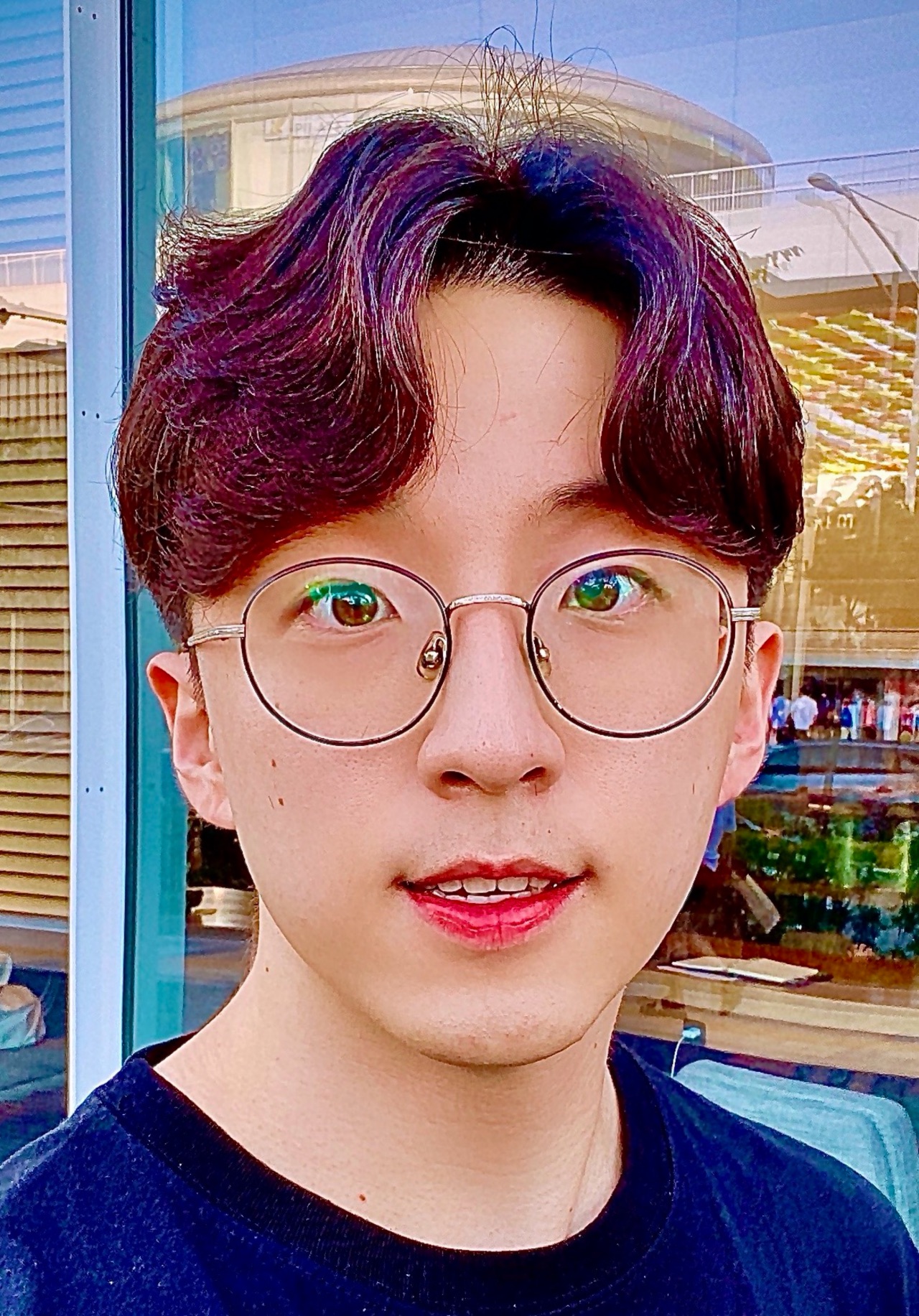}}]{Hong-Bae Jeon}
(S'19) received the B.S degree in electrical and electronic engineering and mathematics from Yonsei University, South Korea, in 2017. He is currently pursuing the Ph.D degree in the School of Integrated Technology, Yonsei University, South Korea. His research interests include applied mathematics and emerging technologies for 5G/B5G communications, such as reconfigurable intelligent surface, UAV network, radio resource management, and free-space optical communications.
\end{IEEEbiography}	
\begin{IEEEbiography}
    [{\includegraphics[width=1in,height=1.25in,clip,keepaspectratio]{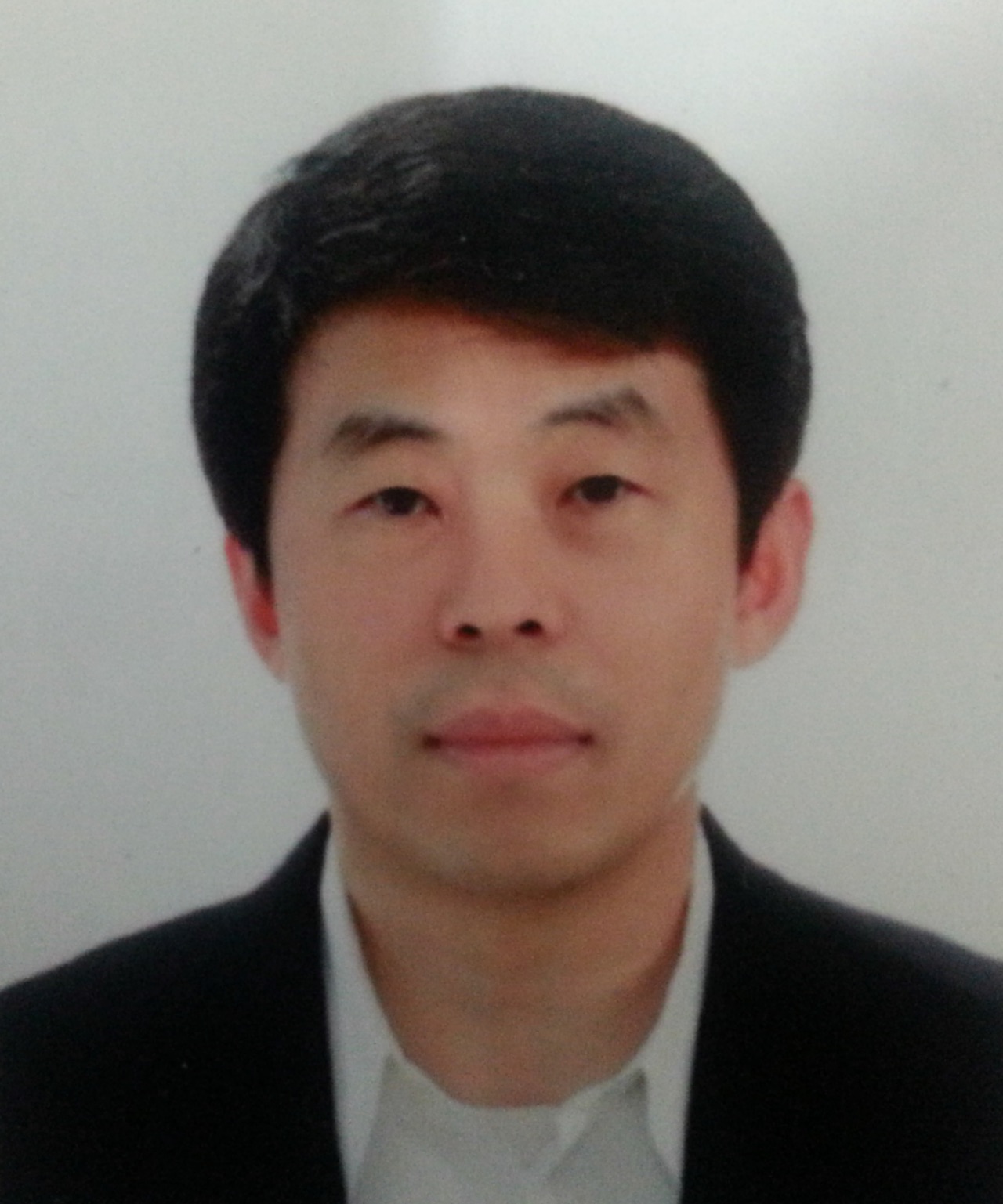}}]{Sung-Ho Park}
(M'87) received the B.S. degree in mathematics from Korea University, Seoul, South Korea, in 1984. From 1985 to 2005, he was working in a software company, dealing with word-processor, middle-ware, and software for financial industry. Since 2005, he has been developing spectrum management and analysis software for government and military. He is currently developing spectrum analysis software for a military operation.
\end{IEEEbiography}	
	\begin{IEEEbiography}
    [{\includegraphics[width=1in,height=1.25in,clip,keepaspectratio]{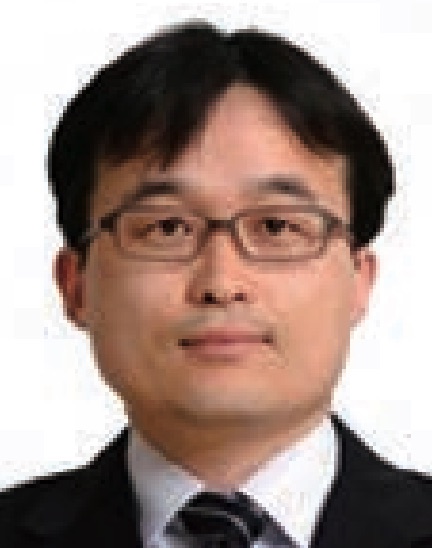}}]{Jaedon Park}
received the B.S. degree in electronics engineering from Hanyang University, Seoul, South Korea, in 2000, and the M.S. and Ph.D. degrees from the School of Electrical Engineering, Korea Advanced Institute of Science and Technology (KAIST), Daejeon, South Korea, in 2002 and 2016, respectively. He is currently a Senior Researcher with the Agency for Defense Development (ADD), Daejeon. His research interests include MIMO systems, relay systems, and FSO systems.
\end{IEEEbiography}	
	\begin{IEEEbiography}
    [{\includegraphics[width=1in,height=1.25in,clip,keepaspectratio]{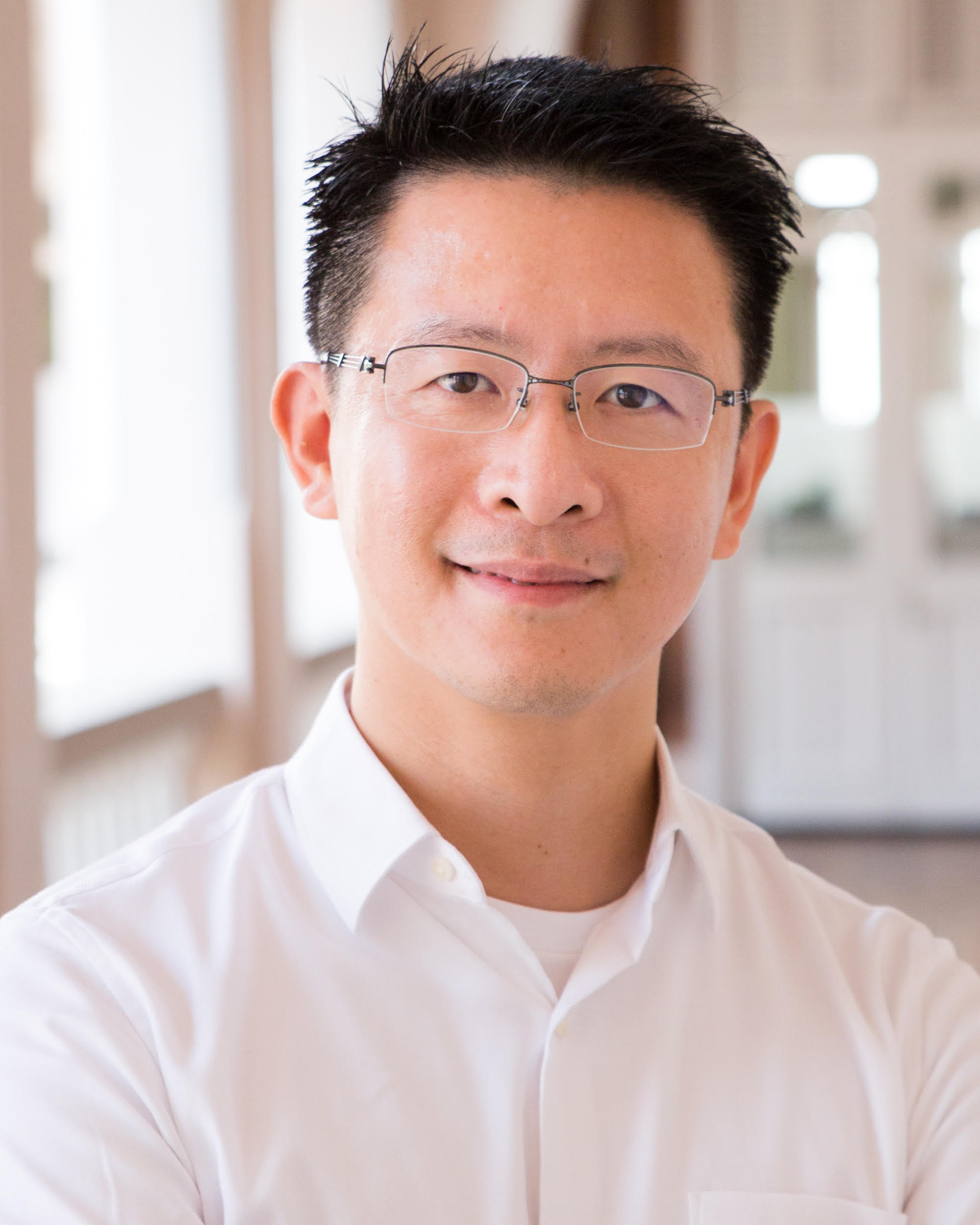}}]{Kaibin Huang}
(M'08-SM'13-F'21) received the B.Eng. and M.Eng. degrees from the National University of Singapore and the Ph.D. degree from The University of Texas at Austin, all in electrical engineering. He is currently an Associate Professor with the Department of Electrical and Electronic Engineering, The University of Hong Kong, Hong Kong. He received the IEEE Communication Society’s 2021 Best Survey Paper, 2019 Best Tutorial Paper, 2019 Asia–Pacific Outstanding Paper, 2015 Asia–Pacific Best Paper Award, and the best paper awards at IEEE GLOBECOM 2006 and IEEE/CIC ICCC 2018. He received the Outstanding Teaching Award from Yonsei University, South Korea, in 2011. He has been named as a Highly Cited Researcher by the Clarivate Analytics in 2019 and 2020. He served as the Lead Chair for the Wireless Communications Symposium of IEEE GLOBECOM 2017 and the Communication Theory Symposium of IEEE GLOBECOM 2014, and the TPC Co-chair for IEEE PIMRC 2017 and IEEE CTW 2013. He is also an Executive Editor of IEEE TRANSACTIONS ON WIRELESS COMMUNICATIONS, an Associate Editor of IEEE JOURNAL ON SELECTED AREAS IN COMMUNICATIONS, and an Area Editor for IEEE TRANSACTIONS ON GREEN COMMUNICATIONS AND NETWORKING. Previously, he served on the Editorial Board for IEEE WIRELESS COMMUNICATIONS LETTERS. He has guest edited special issues of IEEE JOURNAL ON SELECTED AREAS IN COMMUNICATIONS, IEEE JOURNAL OF SELECTED TOPICS IN SIGNAL PROCESSING, and IEEE COMMUNICATIONS MAGAZINE. He is also a Distinguished Lecturer of the IEEE Communications Society and the IEEE Vehicular Technology Society, and a Research Fellow of Hong Kong Research Grants Council.
\end{IEEEbiography}
	\begin{IEEEbiography}
    [{\includegraphics[width=1in,height=1.25in,clip,keepaspectratio]{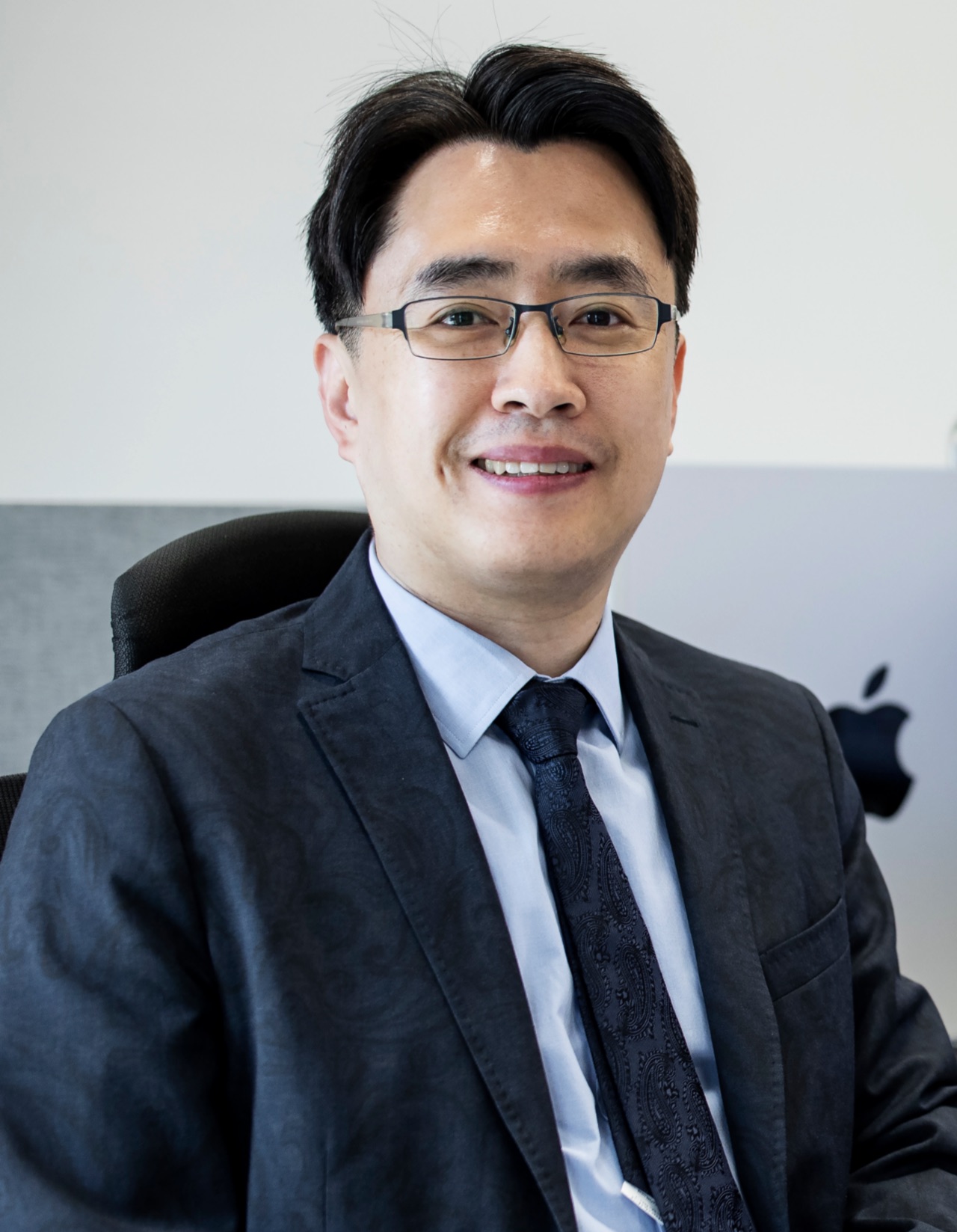}}]{Chan-Byoung Chae}
(S'06-M'09-SM'12-F'21) received the Ph.D. degree in electrical and computer engineering from The University of Texas (UT) at Austin, in 2008. From 2001 to 2005, he was a Research Engineer with the Telecommunications Research and Development Center, Samsung Electronics, Suwon, South Korea. From 2008 to 2009, he was a Postdoctoral Research Fellow with Harvard University, Cambridge, MA, USA. From 2009 to 2011, he was MTS with Bell Labs, Alcatel-Lucent, Murray Hill, NJ, USA. He was a member of the Wireless Networking and Communications Group (WNCG) at UT. He is currently an Underwood Distinguished Professor with the School of Integrated Technology, Yonsei University, South Korea. 

Dr. Chae was a recipient/co-recipient of the IEEE WCNC Best Demo Award in 2020, the Young Engineer Award from the National Academy of Engineering of Korea (NAEK) in 2019, the IEEE DySPAN Best Demo Award, in 2018, the IEEE/KICS JOURNAL OF COMMUNICATIONS AND NETWORKS Best Paper Award in 2018, the Award of Excellence in Leadership of 100 Leading Core Technologies for Korea 2025 from the NAEK in 2017, the Yonam Research Award from the LG Yonam Foundation, in 2016, the IEEE INFOCOM Best Demo Award in 2015, the IEIE/IEEE Joint Award for Young IT Engineer of the Year in 2014, the KICS Haedong Young Scholar Award in 2013, the IEEE Signal Processing Magazine Best Paper Award in 2013, the IEEE ComSoc AP Outstanding Young Researcher Award, in 2012, the IEEE VTS Dan. E. Noble Fellowship Award, in 2008. He is currently the Editor-in-Chief of the IEEE TRANSACTIONS ON MOLECULAR, BIOLOGICAL, AND MULTI-SCALE COMMUNICATIONS and a Senior Editor of the IEEE WIRELESS COMMUNICATIONS LETTERS. He has been serving as an Editor for the IEEE COMMUNICATIONS MAGAZINE, since 2016 and  has worked with the IEEE TRANSACTIONS ON WIRELESS COMMUNICATIONS from 2012 to 2017. He is an IEEE ComSoc Distinguished Lecturer for the term 2020–2021.
\end{IEEEbiography}	
	
\end{document}